\documentclass[11pt,reqno]{amsart}

\usepackage[T1]{fontenc}
\usepackage[utf8]{inputenc}
\usepackage{mathpazo}
\usepackage{amsmath,amssymb,amsthm,mathtools}
\usepackage{stmaryrd}
\usepackage[margin=1.15in]{geometry}
\usepackage{graphicx}
\usepackage{microtype}
\usepackage{enumitem}
\usepackage{booktabs}
\usepackage{url}
\usepackage{tikz}
\usetikzlibrary{arrows.meta,positioning,fit,calc}
\usepackage{xcolor}
\definecolor{IndigoLink}{RGB}{40,70,120}
\definecolor{PlumAccent}{RGB}{122,59,78}
\definecolor{PineAccent}{RGB}{47,94,78}
\usepackage[colorlinks=true,linkcolor=IndigoLink,citecolor=IndigoLink,
            urlcolor=IndigoLink]{hyperref}
\usepackage{aliascnt}
\hypersetup{pdftitle={Completeness and incompleteness of basic matching logic},
            pdfauthor={Xiaohong Chen and Grigore Rosu}}
\usepackage[nameinlink]{cleveref}

\newcommand{\Var}{\mathrm{Var}}
\newcommand{\SVar}{\mathrm{SVar}}
\newcommand{\FV}{\mathrm{FV}}
\newcommand{\lloc}{\vDash_{\mathrm{loc}}}

\newcommand{\den}[1]{\llbracket #1\rrbracket}
\newcommand{\out}{{*}_{M\setminus C}}
\newcommand{\bx}[1]{[#1]}
\newcommand{\dm}[1]{\langle #1\rangle}
\newcommand{\Pw}{\mathcal{P}}
\newcommand{\ii}[2]{\overline{#1}(#2)}
\newcommand{\ldi}[1]{\mathrel{\rightsquigarrow_{#1}}}
\newcommand{\ld}[1]{\mathord{\rightsquigarrow_{#1}}}
\newcommand{\ldn}{\mathord{\rightsquigarrow}}
\newcommand{\lds}{\mathord{\rightsquigarrow^{*}}}
\newcommand{\im}[2]{\ld{#1}[#2]}
\newcommand{\ims}[1]{\lds[#1]}
\DeclareMathOperator{\dep}{d}
\newcommand{\app}{\mathbin{\cdot}}

\theoremstyle{plain}
\newtheorem{theorem}{Theorem}
\newaliascnt{proposition}{theorem}
\newtheorem{proposition}[proposition]{Proposition}
\aliascntresetthe{proposition}
\newaliascnt{lemma}{theorem}
\newtheorem{lemma}[lemma]{Lemma}
\aliascntresetthe{lemma}
\newaliascnt{corollary}{theorem}
\newtheorem{corollary}[corollary]{Corollary}
\aliascntresetthe{corollary}
\theoremstyle{definition}
\newaliascnt{definition}{theorem}
\newtheorem{definition}[definition]{Definition}
\aliascntresetthe{definition}
\newaliascnt{remark}{theorem}
\newtheorem{remark}[remark]{Remark}
\aliascntresetthe{remark}
\crefname{theorem}{Theorem}{Theorems}
\crefname{proposition}{Proposition}{Propositions}
\crefname{lemma}{Lemma}{Lemmas}
\crefname{corollary}{Corollary}{Corollaries}
\crefname{definition}{Definition}{Definitions}
\crefname{remark}{Remark}{Remarks}
\Crefname{theorem}{Theorem}{Theorems}
\Crefname{proposition}{Proposition}{Propositions}
\Crefname{lemma}{Lemma}{Lemmas}
\Crefname{corollary}{Corollary}{Corollaries}
\Crefname{definition}{Definition}{Definitions}
\Crefname{remark}{Remark}{Remarks}

\newenvironment{ednote}
  {\par\medskip\noindent\begingroup\small\color{black!78}
   \setlength{\leftskip}{1.1em}
   \makebox[0pt][r]{\textcolor{PlumAccent}{\rule[-.2ex]{1.5pt}{1.6ex}}\hspace{.6em}}%
   \textsc{\textcolor{PlumAccent}{Note.}}\ \ignorespaces}
  {\par\endgroup\medskip}

\begin{document}

\title{Completeness and incompleteness of basic matching logic}
\author{Xiaohong Chen (Intent Computing, Inc.)\\
Grigore Ro\c{s}u (University of Illinois Urbana-Champaign)}
\thanks{\emph{Acknowledgements.} Claude Opus 5 (Anthropic) produced the
constructions and proofs presented here, including the double cover of
\Cref{sec:B}, the reduction of \Cref{sec:mu}, and the obstruction of
\Cref{sec:limits}. It also drafted the text under our direction and subject to
our verification and correction. We posed the problem, supplied the earlier
results on which it rests, and remain accountable for the work. We thank
ChatGPT Sol for a sequence of detailed referee reports that identified errors
in earlier versions of \Cref{sec:mu} and the many-sorted argument. We have
adopted the corrections proposed in those reports.}
\date{August 9, 2026}

\begin{abstract}\noindent
\emph{Basic} matching logic is matching logic \emph{without definedness}.
Symbols are interpreted as set-valued operations, element variables denote
singletons and are bound by $\exists$, and no connective in the language
uniformly internalizes totality.
For basic matching logic without fixpoints over an arbitrary one-sorted
finitary signature, we prove \emph{global completeness}
($\Gamma\vDash\varphi$ iff $\Gamma\vdash\varphi$,
for arbitrary and possibly infinite $\Gamma$), and obtain as a corollary
that the definedness extension is conservative. The proof \emph{localizes}
$\Gamma$ to a theory $\Delta_\Gamma$ and reduces both global semantic
consequence and syntactic derivability to the same local consequence relation,
$\Gamma\vDash\varphi$ iff $\Delta_\Gamma\lloc\varphi$ iff
$\Gamma\vdash\varphi$. A double-cover construction establishes the semantic
equivalence.
Adding least fixpoints destroys effective axiomatizability. Over a signature
with one unary
and two binary symbols and no constants, validity is not recursively
enumerable, so no sound calculus with a recursively enumerable proof relation
is even \emph{weakly} complete, already for the empty theory and without
definedness. The positive result is also sharp with respect to sorts. Global
completeness fails with three sorts for a satisfiable $\Gamma$. Together, the two halves
settle the conjecture of \cite{ChenRosu2019}, affirmatively for one sort and
negatively in general. The negative results arise from sort flow, effectivity
with fixpoints, and, for hybrid logic
\cite{BlackburnSeligman1995,ArecesTenCate2007}, an obstruction that defeats
every well-founded calculus whose leaves are hypotheses or valid patterns and
whose rules respect the localization.
The last mechanism yields a dichotomy that can be stated without reference to
matching logic. The language with state variables bound by $\exists$ and $\forall$ over
modalities of arbitrary arity is globally complete without nominals, and no
calculus in the well-founded class of that criterion is globally complete once
nominals are added.
\end{abstract}

\maketitle
\markboth{X.\,Chen and G.\,Ro\c{s}u}%
         {Completeness and incompleteness of basic matching logic}

%% ================================================================
\section{Introduction}\label{sec:intro}

Matching logic was introduced through work on the formal semantics of real
programming languages, without a modal-logic framing
\cite{Rosu2010,Rosu2017,ChenRosu2020AML}. In this setting,
a language definition is a theory $\Gamma$ containing the configuration
structure, datatypes and transition rules. The properties to be established
about a program are consequences of $\Gamma$, and few results of interest are
validities of the empty theory. The relevant notion is therefore
\emph{global} consequence, $\Gamma\vDash\varphi$, which quantifies over the
models in which every pattern in $\Gamma$ is total. Each such model is a
realization, or implementation, of the language. Local consequence would tie
proved properties to particular points of particular realizations and hence to
particular implementations. Because implementations can be incorrect, this is
inappropriate for programming-language semantics. Global consequence
distinguishes the realizations that satisfy the language definition from those
that do not. A property proved from $\Gamma$ holds in every realization that
satisfies the definition.

\begin{figure}[!t]
\centering
\hyphenpenalty=10000 \exhyphenpenalty=10000
\resizebox{\textwidth}{!}{%
\begin{tikzpicture}[font=\small,>={Stealth[length=1.7mm]},
  L/.style={align=center,inner sep=2pt},
  ann/.style={font=\scriptsize,align=center,text width=3.3cm,inner sep=1pt},
  annw/.style={font=\scriptsize,align=center,text width=4.3cm,inner sep=1pt},
  annf/.style={font=\scriptsize,align=center,inner sep=1pt},
  lab/.style={font=\scriptsize,align=center,black!60,inner sep=1.5pt},
  a/.style={->,gray!70},
  box/.style={draw=black!55,rounded corners=2.5pt,inner sep=1.4mm},
  boxd/.style={draw=black!55,densely dotted,rounded corners=2.5pt,inner sep=1.4mm}]

  \def\cI{-7.1}\def\cII{-2.55}\def\cIII{1.85}\def\cIV{6.35}
  \def\rI{8.0}\def\rII{5.4}\def\rIII{2.7}\def\rIV{0}

  % row 1
  \node[L] (df) at (\cI,\rI)   {matching logic,\\ with definedness};
  \node[annf,below=0.6mm of df] (dfa)
    {\textcolor{PineAccent}{global completeness \cite{Rosu2017,ChenRosu2019}}\\
     equivalent to $\mathrm{FOL}_=$ \cite{Rosu2017}\\
     equivalent to $\mathrm{H}(@,\forall)$ \cite{LeusteanMoangaSerbanuta2019}};
  \node[L] (dm) at (\cIII,\rI) {matching $\mu$-logic,\\ with definedness};
  \node[annf,PlumAccent,below=0.6mm of dm] (dma)
    {global completeness fails \cite[Prop.~23]{ChenRosu2019};\\
     no effective weak completeness (\Cref{rem:mudef})};

  % row 2
  \node[L] (ms) at (\cIII,\rII) {matching logic,\\ basic many-sorted};
  \node[annf,PlumAccent,below=0.6mm of ms] (msa)
    {global completeness fails\\ (\Cref{sec:regimes})};

  % row 3
  \node[L] (hat) at (\cI,\rIII)   {$\mathrm{H}(@)$, $\mathrm{H}(@,\downarrow)$};
  \node[annf,PineAccent,below=0.6mm of hat] (hata)
    {global completeness \cite{ArecesTenCate2007}};
  \node[L] (bf) at (\cIII,\rIII)  {matching logic,\\ basic unsorted};
  \node[annf,PineAccent,below=0.6mm of bf] (bfa)
    {global completeness\\ (\Cref{cor:main})};
  \node[L] (bm) at (\cIV,\rIII)   {matching $\mu$-logic,\\ basic unsorted};
  \node[annf,PlumAccent,below=0.6mm of bm] (bma)
    {no effective weak\\ completeness (\Cref{thm:mu})};

  % row 4
  \node[L] (K) at (\cIII,\rIV)  {modal logic,\\ basic polyadic};
  \node[ann,PineAccent,below=0.6mm of K] (Ka)
    {global completeness \cite[Ex.~1.5.3]{BlackburnDeRijkeVenema2001}};
  \node[L] (Km) at (\cIV,\rIV)  {modal $\mu$-calculus};
  \node[ann,below=0.6mm of Km,text width=4.4cm] (Kma)
    {\textcolor{PineAccent}{global completeness, finite $\Gamma$,}\\
     \textcolor{PineAccent}{via \cite{Kozen1983,Walukiewicz2000};}
     \textcolor{PlumAccent}{not compact, so}\\
     \textcolor{PlumAccent}{local completeness fails \cite{Kozen1983}}};

  %% frames: results new in this paper
  \node[box,fit=(bf)(bfa)] (bfbox) {};
  \node[box,fit=(ms)(msa)] (msbox) {};
  \node[box,fit=(bm)(bma)] (bmbox) {};
  \node[boxd,fit=(dm)(dma)] (dmbox) {};

  %% H(forall) sits midway between H(@) and the cell on its right,
  %% so that its frame has the same clearance on both sides
  \path let \p1=(hata.east), \p2=(bfbox.west) in
    node[L] (hall) at ({(\x1+\x2)/2-1.0cm},\rIII+0.16) {$\mathrm{H}(\forall)$};
  \node[annf,below=0.6mm of hall] (halla)
    {\textcolor{PineAccent}{local completeness \cite{BlackburnTzakova1998}};\\
     \textcolor{PlumAccent}{no complete calculus in the}\\
     \textcolor{PlumAccent}{class of \Cref{thm:obstruction}}};
  \node[boxd,fit=(hall)(halla)] (hallbox) {};

  %% naming power: upward and leftward.  Arrows run from a node's north edge to
  %% the south edge of the target's annotation, so nothing crosses text.
  \draw[a] (K.north)  -- node[lab,right=1.2mm] {$\Var,\exists$} ([yshift=-1.2mm]bfbox.south);
  \draw[a] ([yshift=1.2mm]bfbox.north) -- node[lab,right=1.2mm] {sorts} ([yshift=-1.2mm]msbox.south);
  \draw[a] (Km.north) -- node[lab,right=1.2mm] {$\Var,\exists$} ([yshift=-1.2mm]bmbox.south);
  \draw[a,shorten <=3.5mm,shorten >=3.5mm] (bfbox.west)  -- node[lab,pos=.5,above=0.8mm] {nominals} (bfbox.west -| hallbox.east);
  \draw[a] (K.west) to[out=180,in=270,looseness=0.7] node[lab,pos=.45,below=2mm] {nominals, $@$} (hata.south);
  \draw[a] (hat.north)      -- node[lab,right=1.2mm] {$\Var,\exists$} (dfa.south);
  \draw[a] ([yshift=1.2mm]hallbox.north) -- node[lab,pos=.5,right=1.2mm] {$@$} ([xshift=9mm]dfa.south);
  \draw[a] ([xshift=-1mm,yshift=1mm]bfbox.north west)  -- node[lab,pos=.55,above=1mm] {$\lceil\cdot\rceil$} ([xshift=20mm,yshift=1.5mm]dfa.south);
  \draw[a] ([xshift=-2.4mm,yshift=1.2mm]msbox.west) -- node[lab,pos=.5,above=1mm] {$\lceil\cdot\rceil$} ([xshift=21mm,yshift=6mm]dfa.south);
  \draw[a] ([yshift=1.2mm]bmbox.north)       to[out=90,in=0] node[lab,pos=.5,right=1.5mm] {$\lceil\cdot\rceil$} ([xshift=1.2mm]dmbox.east);

  %% fixpoints: rightward
  \draw[a,shorten <=3mm,shorten >=3mm] (K.east)  -- node[lab,above=0.4mm] {$\mu$} (Km.west);
  \draw[a,shorten <=2.2mm,shorten >=2.2mm] ([xshift=1mm]bfbox.east) -- node[lab,above=0.4mm] {$\mu$} ([xshift=-1mm]bmbox.west);
  \draw[a] (df.east) -- node[lab,above=0.4mm] {$\mu$} ([xshift=-1.2mm]dmbox.west |- dm.west);
\end{tikzpicture}}
\caption{Matching logic and its modal and hybrid neighbours. Arrows add
language constructs, each node records what is known about completeness there,
and framed nodes carry results proved in this paper.}
\label{fig:landscape}
\end{figure}
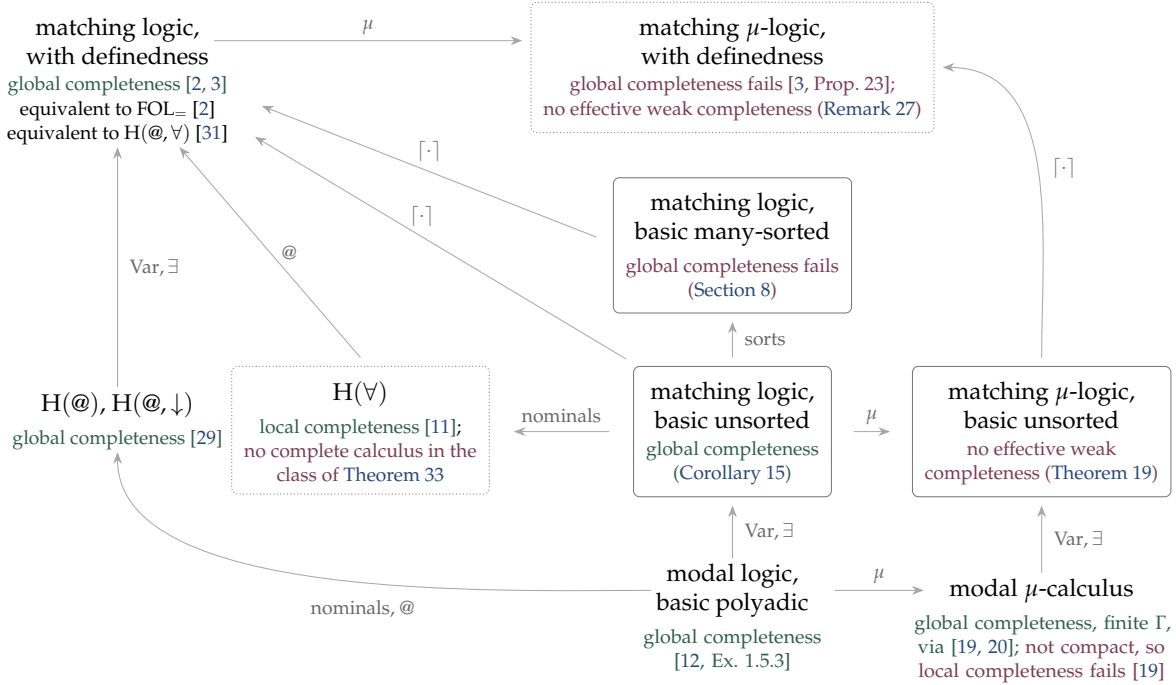

This design allows reasoning principles usually left at the meta-level to be
proved as theorems within the logic. For example, \cite{ChenLucanuRosu2026} axiomatizes initial algebra semantics
as matching logic theories, using a least fixpoint for no-junk
together with no-confusion and equality modulo a set of equations, and then
derives structural induction, iteration and primitive recursion as theorems of
the proof system. In the usual treatments, these are external rules governing
the specification from outside. Definedness, equality, membership, sorts,
functions, constructors and binders are likewise theories internal to the
logic \cite{Rosu2017,ChenLucanuRosu2021,ChenRosu2020ICFP}. They require no
extension of the logic. A single fixed logic and proof system therefore suffice,
with domain-specific apparatus such as a particular programming language or
computational model expressed axiomatically within them. This design raises the
question of the strength of $\vdash$ relative to $\vDash$.

Definedness separates the cases considered in this paper. It is a symbol
$\lceil\cdot\rceil$, axiomatized so that $\lceil\varphi\rceil$ is
total when $\varphi$ is matched by at least one element and empty otherwise.
Its dual $\lfloor\varphi\rfloor=\neg\lceil\neg\varphi\rceil$ says that
$\varphi$ is total, so a language with definedness can speak about the model as
a whole. We call matching logic without definedness \emph{basic}. The word
carries its established modal sense, namely the language with its symbols and
the Boolean connectives and nothing else
\cite[Def.~1.9]{BlackburnDeRijkeVenema2001}, generalized in the same source to
modalities of arbitrary arity \cite[Def.~1.11]{BlackburnDeRijkeVenema2001},
which our polyadic signatures give. Sorts and fixpoints are named separately.
Our completeness theorem concerns basic matching logic with one sort, and
\Cref{sec:mu} shows the incompleteness of basic matching $\mu$-logic, also
unsorted. In the many-sorted case considered in \Cref{sec:regimes}, global
completeness fails.

\Cref{fig:landscape} places these logics among their neighbours, where they appear as
\emph{matching logic, basic unsorted}, as \emph{matching $\mu$-logic, basic
unsorted} and as \emph{matching logic, basic many-sorted}. The arrows indicate
added language constructs. Upward arrows add naming and quantification power,
and rightward arrows add fixpoints. Each node records either the strongest
completeness statement known to hold there or the strongest one known to fail,
using the three senses separated in \Cref{sec:system}.
Framed nodes carry results proved in this paper, and a dotted frame marks a
node where only part of the entry is new. On the hybrid side, $\mathrm{H}(@)$
is modal logic with \emph{nominals}, model-fixed constants denoting singletons,
together with the satisfaction operator $@_i\varphi$, read ``$\varphi$ holds at
the point named $i$''. The logic $\mathrm{H}(@,\downarrow)$ adds the binder
$\downarrow x.\,\varphi$, read ``name the current point $x$''.
By contrast, $\mathrm{H}(\forall)$ quantifies over state variables
\cite{BlackburnSeligman1995,BlackburnTzakova1998,ArecesTenCate2007}. In the
fixpoint-free part of the figure, every branch of the naming axis converges on
\emph{matching logic, with definedness}, and that node has exactly the
expressive power of first-order logic with equality. It translates into pure
predicate logic with equality and, conversely, contains it as an instance
\cite[\S10]{Rosu2017}. On the hybrid side, adding the ``somewhere'' modality to
$\exists$ reaches that same first-order power. The operator $\exists$ alone does
not \cite[Thm.~4.1, Prop.~4.11]{BlackburnSeligman1995}. The identification of
this node with $\mathrm{H}(@,\forall)$ is given in
\cite{LeusteanMoangaSerbanuta2019}.
Its $\mu$ counterpart to the right lies beyond first-order logic, since $\mu$
is not first-order expressible. The definedness-free region is the one at issue
here.

The Hilbert system of \cite{ChenRosu2019} is sound with or without fixpoints.
Every completeness statement recalled below concerns the \emph{fixpoint-free}
system. By \Cref{thm:mu}, none could hold in the presence of $\mu$. The
fixpoint-free Hilbert system of \cite{ChenRosu2019} is also complete for the
empty theory, which is weak completeness in the sense of \Cref{sec:system}.
This is Theorem 16 of the same paper, proved by a canonical-model construction.

Previous global completeness results require theories that contain definedness
\cite{Rosu2017}, where equality and membership are available as derived
notations, $\varphi=\psi$ abbreviating
$\lfloor\varphi\leftrightarrow\psi\rfloor$ and $x\in\varphi$ abbreviating
$\lceil x\wedge\varphi\rceil$. The statement is Theorem 15 of
\cite{ChenRosu2019}. Indeed, $\lfloor\cdot\rfloor$ is the universal modality, so global
consequence reduces to local consequence in one step
\cite{GorankoPassy1992}. Whether the \emph{definedness-free} system is globally
complete was left open in \cite[\S IX-B]{ChenRosu2019}, where it was
conjectured that it is, and again in \cite[Ch.~3]{Chen2023Thesis}, where it is
posed as an open problem alongside the definedness and local completeness
theorems. We prove the conjecture for one sort and disprove it in general.

Modal logic provides two standard routes to global completeness. Both fail in
our setting. The first route reduces global consequence to local consequence. For basic
modal logic, $\Gamma\vDash_{\mathrm{glo}}\varphi$ iff
$\Box^\omega\Gamma\vDash_{\mathrm{loc}}\varphi$
\cite[Ex.~1.5.3]{BlackburnDeRijkeVenema2001}. This semantic statement yields no
derivation by itself. However, every $\Box^n\gamma$ is derivable from $\Gamma$
by necessitation, and strong local completeness turns the right-hand side into
a derivable finite implication. Two applications of modus ponens then give
$\Gamma\vdash\varphi$. The standard proof of this reduction fails in our
setting. It starts from a point $w$ at which all the boxed hypotheses hold,
passes to the submodel generated by $w$, notes that $\Gamma$ is globally true
in that submodel, and transfers the resulting truth of $\varphi$ back to $w$.
This last step uses the invariance of modal formulas under generated
submodels. Matching logic patterns are not invariant under generated submodels,
because $\exists x$ ranges over the whole carrier, including the points that
the generated submodel discards \cite[\S3]{ArecesTenCate2007}. In hybrid terms,
our language is the nominal-free fragment of $\mathrm{H}(\forall)$
\cite{BlackburnSeligman1995,ArecesTenCate2007}. The second route trades $\Gamma\vDash\varphi$ for the validity of a finite
implication. It needs a deduction theorem, and the Hilbert system of
\cite{ChenRosu2019} has none. Universal generalization applies there to
hypotheses as well as to theorems, so $\{x\}\vdash\forall x.\,x$, while
$x\to\forall x.\,x$ is invalid. In a model with at least two elements, $x$
denotes a singleton, whereas $\forall x.\,x$ denotes the empty set.

We instead \emph{localize} the theory, using the theory $\Delta_\Gamma$ defined
in \Cref{def:boxed} and characterized semantically in \Cref{lem:whatDelta}. We
prove two identities,
\[
  \Gamma\vDash\varphi \iff \Delta_\Gamma\lloc\varphi
  \qquad\text{and}\qquad
  \Gamma\vdash\varphi \iff \Delta_\Gamma\lloc\varphi ,
\]
Together, these identities give global completeness (\Cref{cor:main}) for one
sort and an arbitrary finitary signature. The other two components of the
classical route remain available. Necessitation is \Cref{lem:GN}, and strong
local completeness is available from \cite{ChenRosu2019TR}, which we assume as the
input~(L) of \Cref{sec:system} and use without opening its proof. For each
identity, one direction is direct and the other requires an argument. The
second identity relies on the available local completeness result. For the
first, the double-cover construction developed in
\Cref{lem:locality,lem:twocopy} turns a local countermodel to
$\Delta_\Gamma\lloc\varphi$ into a model of $\Gamma$ that refutes $\varphi$.
It provides a model-theoretic surrogate for the universal modality, which the
language cannot express.

The definedness extension is \emph{conservative} (\Cref{cor:cons}). A
decade of reasoning that freely assumed $\lceil x\rceil$ was therefore justified
in the one-sorted setting proved here.
The result is also sharp in two directions. It fails as soon as there are several
sorts, still with no definedness anywhere, because the conclusion can be placed
at a sort that the localization never reaches. The counterexample of
\Cref{sec:regimes} uses three sorts. It also fails once fixpoints are added,
where by \Cref{thm:mu} validity is not recursively enumerable, so no sound
calculus with a recursively enumerable proof relation is even weakly complete,
already for the empty theory (\Cref{sec:mu,sec:regimes}). The empty-theory result
is new. Failure for a nonempty theory was already known
\cite[Prop.~22--23]{ChenRosu2019}.
Applicative matching logic is the instance with one binary symbol and constant
symbols only, so it is also globally complete (\Cref{sec:aml}). This does not
conflict with the many-sorted failure just mentioned. Sorts, binders and type
systems are encodable in applicative matching logic, but the existing encodings
use definedness. \Cref{cor:noencoding} rules out an analogous definedness-free
encoding that both preserves global consequence and reflects derivability. In
\Cref{sec:limits}, we isolate the assumptions used by the proof of \Cref{thm:A}
and extend both failures from the calculus of \Cref{fig:system} to every
calculus in a broad class.

The same criterion also resolves an apparently open question about hybrid modal
logic that can be stated without reference to matching logic. Consider the hybrid language with state
variables bound by $\exists$ and $\forall$, modalities of arbitrary arity, and
no propositional variables, which is the $\exists$-language studied in
\cite{BlackburnSeligman1995}. The nominal-free language is globally complete for
arbitrary $\Gamma$ (\Cref{cor:main}). To our knowledge, this result was
previously unknown. Adding nominals yields $\mathrm{H}(\forall)$ as presented in
\cite{BlackburnTzakova1998}, for which global completeness fails
(\Cref{cor:hforall}). Two qualifications are important. The positive result
relies on strong local completeness for the nominal-free fragment, which we take
from \cite{ChenRosu2019TR}. The classical construction does not supply this
result because it uses nominals as its Henkin witnesses. Global consequence in
$\mathrm{H}(\forall)$ is recursively enumerable for an effectively presented
countable signature and a recursively enumerable theory. The negative result
therefore does not establish non-axiomatizability, and a complete system
exists. Nevertheless, no calculus in the well-founded class of
\Cref{thm:obstruction} can be complete (\Cref{rem:whatcomplete}). This latter
result may be more surprising because nominals usually repair completeness,
as with the pure extensions of $\mathrm{H}(@)$, which are complete once
\textup{(Name)} and \textup{(Paste)} are available
\cite{BlackburnTenCate2006}. Here we work at the level of models and ask about
global consequence, where a nominal lets $\Gamma$ constrain the
cardinality of the carrier.

\Cref{sec:system} presents the system and the two assumed results.
\Cref{sec:boxed} defines the localization $\Delta_\Gamma$, \Cref{sec:B} proves
the two identities, and \Cref{sec:composite} combines them into
\Cref{cor:main}. \Cref{sec:aml} specializes the result to applicative matching
logic. \Cref{sec:mu} proves incompleteness with fixpoints,
\Cref{sec:regimes} compares the four settings, and \Cref{sec:limits} shows that
the two failures do not depend on the calculus. \Cref{sec:conclusions}
discusses the relation to modal and hybrid logic, mechanization, and open
problems.

%% ================================================================
\section{The system}\label{sec:system}

\subsection*{Syntax}
We work in \emph{basic matching logic, unsorted}, with no definedness and no
fixpoints until \Cref{sec:mu}.
Fix one sort and a set $\Sigma$ of symbols, each with an arity $n\ge 0$.
The set $\Sigma$ may be infinite, and symbols of arity $0$ are \emph{constants}.
Element variables $\Var$ are countably infinite. Primitives are implication
and $\bot$:
\[
  \varphi ::= x \;\mid\; \sigma(\varphi_1,\ldots,\varphi_n) \;\mid\;
              \varphi_1\to\varphi_2 \;\mid\; \bot \;\mid\;
              \exists x.\,\varphi
  \qquad (x\in\Var,\ \sigma\in\Sigma \text{ of arity } n),
\]
with $\neg\varphi:=\varphi\to\bot$, $\top:=\bot\to\bot$,
$\varphi_1\vee\varphi_2:=\neg\varphi_1\to\varphi_2$,
$\varphi_1\wedge\varphi_2:=\neg(\varphi_1\to\neg\varphi_2)$, and
$\forall x.\,\varphi:=\neg\exists x.\,\neg\varphi$. A constant is written
$\sigma$ rather than $\sigma()$. There are no set variables or fixpoints, and
there is no definedness symbol. \Cref{rem:setvars} discusses set variables.

\subsection*{Models and denotations}
A model is a nonempty set $M$ together with, for each $\sigma\in\Sigma$ of
arity $n$, a map $\sigma_M:M^n\to\Pw(M)$ (so $\sigma_M\subseteq M$ when
$n=0$), extended pointwise to sets by
\[
  \sigma_M(A_1,\ldots,A_n)
  := \bigcup\{\sigma_M(a_1,\ldots,a_n) \mid a_i\in A_i\},
\]
which is $\varnothing$ as soon as some $A_i$ is. A valuation is
$\rho:\Var\to M$, and the \emph{denotation} $\ii{\rho}{\varphi}\subseteq M$ is
\[
  \ii{\rho}{x}=\{\rho(x)\},\qquad
  \ii{\rho}{\sigma(\varphi_1,\ldots,\varphi_n)}
    =\sigma_M\bigl(\ii{\rho}{\varphi_1},\ldots,\ii{\rho}{\varphi_n}\bigr),
  \qquad
  \ii{\rho}{\bot}=\varnothing,
\]
\[
  \ii{\rho}{\varphi_1\to\varphi_2}
   =\bigl(M\setminus\ii{\rho}{\varphi_1}\bigr)\cup\ii{\rho}{\varphi_2},
  \qquad
  \ii{\rho}{\exists x.\,\varphi}
   =\bigcup_{a\in M}\ii{\rho[a/x]}{\varphi} .
\]
The derived clauses are $\ii{\rho}{\neg\varphi}=M\setminus\ii{\rho}{\varphi}$,
$\ii{\rho}{\top}=M$, and
$\ii{\rho}{\forall x.\,\varphi}=\bigcap_{a\in M}\ii{\rho[a/x]}{\varphi}$. A
set $\Delta$ of patterns denotes conjunctively,
$\ii{\rho}{\Delta}:=\bigcap_{\delta\in\Delta}\ii{\rho}{\delta}$, with value
$M$ when $\Delta=\varnothing$. When every member of $\Delta$ is closed, this
does not depend on $\rho$, and we write $\den{\Delta}$.

Although absent from the language, the definedness symbol occurs frequently in
the discussion below, so we record its semantics. It is the unary operator
$\lceil\cdot\rceil$, whose semantics is
$\ii{\rho}{\lceil\varphi\rceil}=M$ if $\ii{\rho}{\varphi}\ne\varnothing$ and
$\varnothing$ otherwise, with dual
$\lfloor\varphi\rfloor:=\neg\lceil\neg\varphi\rceil$, which is total exactly
when $\varphi$ is. In modal terms these are the existential and universal
modalities, $\lceil\cdot\rceil=E$ and $\lfloor\cdot\rfloor=A$, where
$A\varphi$ holds at a point iff $\varphi$ holds at every point of the model
\cite{GorankoPassy1992}. Neither is available in the language above.
\Cref{sec:conclusions} shows that they are not even definable in it.

\subsection*{Totality and the three consequence relations}
\begin{definition}\label{def:total}
$\varphi$ is \emph{total} in $M$ under $\rho$ if $\ii{\rho}{\varphi}=M$. Then
\[
\begin{aligned}
  M\vDash\varphi
    &\ :\iff\ \ii{\rho}{\varphi}=M \text{ for every } \rho
      &&\text{($\varphi$ is total in $M$)};\\
  M\vDash\Gamma
    &\ :\iff\ M\vDash\gamma \text{ for every } \gamma\in\Gamma;\\
  \vDash\varphi
    &\ :\iff\ M\vDash\varphi \text{ for every model } M
      &&\text{(\emph{validity})};\\
  \Delta\lloc\varphi
    &\ :\iff\ \ii{\rho}{\Delta}\subseteq\ii{\rho}{\varphi}
      \text{ for every } M,\rho
      &&\text{(\emph{local})};\\
  \Gamma\vDash\varphi
    &\ :\iff\ M\vDash\varphi \text{ for every } M\vDash\Gamma
      &&\text{(\emph{global})}.
\end{aligned}
\]
\end{definition}

The relations $\varnothing\lloc\varphi$, $\varnothing\vDash\varphi$ and
$\vDash\varphi$ all coincide. The relation $\lloc$ compares denotations
pointwise, whereas $\vDash$ asks for totality. This distinction is the source of
the difficulty.

\emph{Closed patterns.} Throughout, we assume without loss of generality that
$\Gamma$ and $\varphi$ are closed. Writing $\widehat\psi$ for the universal closure
of $\psi$, we have $M\vDash\psi$ iff $M\vDash\widehat\psi$, because
$\ii{\rho}{\widehat\psi}=\bigcap_{\bar a}\ii{\rho[\bar a/\bar x]}{\psi}$.
Moreover,
$\Gamma\vdash\psi$ iff $\widehat\Gamma\vdash\widehat\psi$, by universal
generalization in one direction and instantiation in the other. Hence both
$\Gamma\vdash\varphi$ and $\Gamma\vDash\varphi$ are unchanged by closing, and
closing also preserves $\Delta_\Gamma\lloc\varphi$ because
$\den{\Delta_\Gamma}$ will not depend on the valuation (\Cref{def:boxed}). In
general, local consequence does \emph{not} survive closing its conclusion. With
$\Gamma$ and $\varphi$ closed,
\[
  M\vDash\Gamma \iff \den{\Gamma}=M ,
  \qquad
  M\vDash\varphi \iff \den{\varphi}=M ,
\]
so no relation among the three carries a valuation quantifier. Inside proofs,
however, subpatterns and the intermediate patterns of a derivation are open,
and there $\ii{\rho}{\cdot}$ is unavoidable.

\subsection*{Convention}
\Cref{def:total} is the only place where models are quantified over
explicitly. Everywhere else a model $M$ is fixed but arbitrary and left
implicit. The interpretations $\sigma_M$, $\ii{\rho}{\varphi}$ and the
reachability relations of \Cref{sec:boxed} are all taken in that $M$, and a
claim stated without a model is a claim about every $M$. The exception is
\Cref{sec:B}, where a second model $N$ is built alongside $M$. From the point
where $N$ appears, we index the
denotation by the model, writing $\den{\psi}_M$ and $\den{\psi}_N$, and
likewise $\ii{\rho}{\psi}_M$ and $\ii{\nu}{\psi}_N$.

\subsection*{Proof system}
$\Gamma\vdash\varphi$ means that $\varphi$ occurs in a finite sequence each of
whose items is an instance of an axiom scheme of \Cref{fig:system}, a member of
$\Gamma$, or the result of applying one of its inference rules to earlier
items. The system is that of \cite{ChenRosu2019}, restricted to one sort. The
presentation over $\to,\bot$, with the other connectives as derived notations,
follows \cite[Figs.~1--2]{ChenLucanuRosu2026}.

\begin{figure}[t]
\setlength{\fboxsep}{10pt}
\noindent\fbox{\begin{minipage}{.93\textwidth}
\small
\emph{Application contexts.}
$C ::= \square \mid \sigma(\varphi_1,\ldots,C,\ldots,\varphi_n)$, where
$\square$ is the hole. The notation $C[\varphi]$ plugs $\varphi$ into it.
\smallskip\hrule\smallskip
\begin{enumerate}[leftmargin=2.9em,itemsep=.35em,label=\textbf{(\arabic*)}]
\item \textbf{Tautology.} $\vdash\varphi$, for $\varphi$ a substitution
  instance of a propositional tautology over $\to$ and $\bot$ (patterns
  substituted for the propositional atoms).
\item \textbf{Modus ponens.} From $\varphi_1$ and $\varphi_1\to\varphi_2$
  infer $\varphi_2$.
\item \textbf{$\exists$-quantifier.}
  $\vdash\varphi[y/x]\to\exists x.\,\varphi$, capture-avoiding.
\item \textbf{$\exists$-generalization.} From $\varphi_1\to\varphi_2$ infer
  $(\exists x.\,\varphi_1)\to\varphi_2$, provided $x\notin\FV(\varphi_2)$.
\end{enumerate}
\smallskip
For every $\sigma\in\Sigma$ of arity $n\ge 1$ and $1\le i\le n$, writing
$\vec\varphi$ for the remaining arguments:
\begin{enumerate}[leftmargin=2.9em,itemsep=.35em,start=5,
                  label=\textbf{(\arabic*)}]
\item \textbf{Propagation$_\bot$.}
  $\vdash\sigma(\ldots,\bot,\ldots)\to\bot$.
\item \textbf{Propagation$_\vee$.}
  $\vdash\sigma(\ldots,\varphi_1\vee\varphi_2,\ldots)
   \to\sigma(\ldots,\varphi_1,\ldots)\vee\sigma(\ldots,\varphi_2,\ldots)$.
\item \textbf{Propagation$_\exists$.}
  $\vdash\sigma(\ldots,\exists x.\,\varphi,\ldots)
   \to\exists x.\,\sigma(\ldots,\varphi,\ldots)$,
  provided $x\notin\FV(\vec\varphi)$.
\item \textbf{Framing.} From $\varphi_1\to\varphi_2$ infer
  $\sigma(\ldots,\varphi_1,\ldots)\to\sigma(\ldots,\varphi_2,\ldots)$.
\end{enumerate}
\smallskip
\begin{enumerate}[leftmargin=2.9em,itemsep=.35em,start=9,
                  label=\textbf{(\arabic*)}]
\item \textbf{Existence.} $\vdash\exists x.\,x$.
\item \textbf{Singleton variable.}
  $\vdash C_1[x\wedge\varphi]\to\neg\,C_2[x\wedge\neg\varphi]$, for all application
  contexts $C_1,C_2$.
\end{enumerate}
\smallskip\hrule\smallskip
\emph{Derived.} Universal generalization: from $\varphi$ infer
$\forall x.\,\varphi$, with no side condition on $x$.
\end{minipage}}
\caption{The proof system. Rules (1)--(4) and (9) form the quantifier logic of
element variables. Rules (5)--(8) state that every symbol is a normal modality
in each argument. Rule (10) states that element variables denote singletons.}
\label{fig:system}
\end{figure}

\subsection*{The two inputs}
\begin{itemize}[leftmargin=1.4em,itemsep=.2em]
\item[(L)] \emph{Strong local completeness.} If $\Delta\lloc\varphi$ then
  $\vdash\bigl(\bigwedge\Delta_0\bigr)\to\varphi$ for some finite
  $\Delta_0\subseteq\Delta$.
\item[(S)] \emph{Soundness.} $\Gamma\vdash\varphi$ implies
  $\Gamma\vDash\varphi$.
\end{itemize}
We use both results as black boxes without opening their proofs. We refer to
(L) below as our \emph{oracle}, the one nontrivial fact about $\vdash$ used by
the argument. (L) is Definition
3.3 and Theorem 3.8 of \cite{Chen2023Thesis} and is also proved in
\cite{ChenRosu2019TR}. In \cite{ChenRosu2019}, it yields Theorem 16, the weak
completeness statement that $\vDash\varphi$ implies $\vdash\varphi$, by taking
$\Delta$ to be empty. See \cite[\S IX-B]{ChenRosu2019} for the three notions of
completeness and their relationship. The result is proved by a canonical-model
construction for many-sorted polyadic signatures, extending the completeness
proof for hybrid logic with binders of \cite{BlackburnTzakova1998} to multiple
sorts and modalities of arbitrary arity. Our~(L) is its one-sorted special case.
(S) is Theorem 13 of \cite{ChenRosu2019}.

That source fixes a countable signature, while we allow an arbitrary finitary
$\Sigma$, so we record why (L) is available in the generality used here.
Suppose $\Delta\lloc\varphi$. Under the standard translation of
\cite[\S10]{Rosu2017} this is a first-order consequence, so by compactness
there is a finite $\Delta_0\subseteq\Delta$ with $\Delta_0\lloc\varphi$. Let
$\Sigma_0$ be the finite set of symbols occurring in $\Delta_0\cup\{\varphi\}$.
Every $\Sigma_0$-model expands to a $\Sigma$-model by interpreting the remaining
symbols arbitrarily, and this leaves the denotation of every
$\Sigma_0$-pattern unchanged, so $\Delta_0\lloc\varphi$ holds over $\Sigma_0$
as well. The cited theorem applies there and yields
$\vdash(\bigwedge\Delta_0')\to\varphi$ for a finite
$\Delta_0'\subseteq\Delta_0$. That derivation is also a derivation over $\Sigma$.

We distinguish three completeness statements. \emph{Weak} completeness is the
empty-theory case, in which $\vDash\varphi$ implies $\vdash\varphi$.
\emph{Local} completeness is~(L), which quantifies over an arbitrary $\Delta$
compared pointwise. \emph{Global} completeness states that
$\Gamma\vDash\varphi$ implies $\Gamma\vdash\varphi$ for arbitrary $\Gamma$.
This is \Cref{cor:main}. Each of the last two specializes to the first at
$\Delta=\varnothing$ and $\Gamma=\varnothing$. For~(L),
$\vdash\top\to\varphi$ gives $\vdash\varphi$. Weak completeness is therefore
the weakest of the three. Refuting it, as \Cref{thm:mu} does, refutes the other
two at once.

The two \emph{derivability} relations also differ. \Cref{lem:GN} closes this gap
for $\Delta_\Gamma$. Write
$\Gamma\vdash_{\mathrm{loc}}\varphi$ when
$\vdash(\bigwedge\Gamma_0)\to\varphi$ for some finite
$\Gamma_0\subseteq\Gamma$, which is the conclusion form of~(L). Then
$\Gamma\vdash_{\mathrm{loc}}\varphi$ implies $\Gamma\vdash\varphi$ by modus
ponens. The converse fails. Take $\Gamma=\{\neg x\}$, as in
\cite[\S3.3]{Chen2023Thesis}. Since $\neg x$ is $x\to\bot$,
$\exists$-generalization gives $\Gamma\vdash(\exists x.\,x)\to\bot$, and
\textup{(Existence)} then gives $\Gamma\vdash\bot$. However,
$\nvdash\neg x\to\bot$ because, in a two-element model, $\neg x$ denotes the
point other than $\rho(x)$, which is not contained in $\varnothing$. The
example uses an open pattern, which is one reason we close $\Gamma$ and
$\varphi$ throughout.

\begin{ednote}\label{note:provenance}
Matching logic, with its patterns, its set-valued interpretation of symbols
and its element variables denoting singletons, is due to Ro\c{s}u
\cite{Rosu2010,Rosu2017}, with \cite{ChenLucanuRosu2021} as the expository
account. The proof system of \Cref{fig:system}, local completeness, global
completeness in the presence of definedness, and failure of the deduction
theorem without definedness are all from \cite{ChenRosu2019}. The applicative
restriction to one binary symbol is \cite{ChenRosu2019AML}, developed as a
foundation for $\mathbb{K}$ in \cite{ChenRosu2020AML} and used for binders and
type systems in \cite{ChenRosu2020ICFP}. The interpretation of definedness as a
universal modality and the resulting reduction of global to local consequence
are established in \cite{GorankoPassy1992}. The $\Box^\omega$ form of that
reduction for basic modal logic is Exercise 1.5.3 of
\cite{BlackburnDeRijkeVenema2001}. That matching logic and hybrid logic are two
presentations of one logic under global deduction is
\cite{LeusteanMoangaSerbanuta2019}.
\end{ednote}

%% ================================================================
\section{The boxed theory}\label{sec:boxed}

\begin{definition}[coordinates, reachability, backward closure]\label{def:bc}
A \emph{coordinate} is a pair $e=(\sigma,i)$ with $\sigma\in\Sigma$ of arity
$n\ge 1$ and $1\le i\le n$. Let $E$ be the set of coordinates. Constants
therefore have no coordinates. A symbol points back to each of its arguments,
$\ld{e}:=\{(u,a_i)\mid u\in\sigma_M(a_1,\ldots,a_n)\}\subseteq M\times M$, and
$\ldn:=\bigcup_{e\in E}\ld{e}$. For a word $p=e_1\cdots e_m\in E^*$ let
$\ld{p}:=\ld{e_1};\cdots;\ld{e_m}$ be the relational composite, with
$\ld{\varepsilon}:=\mathrm{id}_M$. Then $\lds=\bigcup_{p\in E^*}\ld{p}$. For a
relation $R$ on $M$ write $R[U]:=\{v\mid u\mathbin{R}v \text{ for some }
u\in U\}$, abbreviating $R[\{u\}]$ to $R[u]$. Finally, $C\subseteq M$ is
\emph{backward closed} if $\ldn[C]\subseteq C$. Equivalently, if $u\in C$ and
$u\in\sigma_M(a_1,\ldots,a_n)$, then $a_1,\ldots,a_n\in C$.
\end{definition}

Thus $\ims{u}$ is the set of points reachable from $u$ by descending into
arguments, and the least backward-closed set containing $u$.

\begin{definition}[boxes]\label{def:boxes}
For $e=(\sigma,i)$ put $\dm{e}\psi:=\sigma(\top,\ldots,\psi,\ldots,\top)$ with
$\psi$ in position $i$, and put $\bx{e}\psi:=\dm{e}(\psi\to\bot)\to\bot$. For
a word $p=e_1\cdots e_m$, $\bx{p}\psi:=\bx{e_1}\cdots\bx{e_m}\psi$, with
$\bx{\varepsilon}\psi:=\psi$.
\end{definition}

\begin{ednote}
This is where the assumption of a single sort is used. Every coordinate
accepts and returns a pattern of that sort, so \emph{every} word in $E^*$
composes and every $\bx{p}\psi$ is well formed. With several sorts a word
from the sort of a hypothesis to the sort of the conclusion exists only along
a chain of symbol inputs, and there may be none, in which case \Cref{thm:B}
fails. See \Cref{sec:regimes}.
\end{ednote}

\begin{lemma}[box semantics]\label{lem:boxsem}
For closed $\psi$,
$\den{\bx{p}\psi}=\{u\in M \mid \im{p}{u}\subseteq\den{\psi}\}$.
\end{lemma}

\begin{proof}
For a single $e=(\sigma,i)$, since $\den{\top}=M$, the pointwise extension
gives $\den{\dm{e}\neg\psi}=\{u\mid \im{e}{u}\cap\den{\neg\psi}\ne\varnothing\}$,
whose complement is $\{u\mid \im{e}{u}\subseteq\den{\psi}\}$. Words compose.
\end{proof}

\begin{lemma}[necessitation]\label{lem:GN}
If $\Gamma\vdash\psi$ then $\Gamma\vdash\bx{p}\psi$ for every $p\in E^*$.
\end{lemma}

\begin{proof}
It suffices to treat a single coordinate $e=(\sigma,i)$ and iterate along $p$.
The tautology $\psi\to((\psi\to\bot)\to\bot)$ with modus ponens gives
$\Gamma\vdash\neg\psi\to\bot$. Framing in position $i$, with $\top$ in the
other positions, gives
$\Gamma\vdash\dm{e}\neg\psi\to\sigma(\top,\ldots,\bot,\ldots,\top)$. Finally,
propagation of $\bot$ in position $i$ gives
$\vdash\sigma(\top,\ldots,\bot,\ldots,\top)\to\bot$. Composing the two
implications, which is propositional, gives
$\Gamma\vdash\dm{e}\neg\psi\to\bot$, that is, $\Gamma\vdash\bx{e}\psi$
(\Cref{def:boxes}).
\end{proof}

\begin{definition}[localization]\label{def:boxed}
The \emph{localization} of $\Gamma$ is
$\ \Delta_\Gamma := \{\bx{p}\gamma \mid \gamma\in\Gamma,\ p\in E^*\}$.
\end{definition}

The localization remains a set even when $\Sigma$ is infinite. Indeed,
$E\subseteq\Sigma\times\mathbb{N}$ is a set, hence so is
$E^*=\bigcup_k E^k$, and $\Delta_\Gamma$ is the image of
$\Gamma\times E^*$ under $(\gamma,p)\mapsto\bx{p}\gamma$. Its members are
closed, so $\den{\Delta_\Gamma}$ is defined, and $\Delta_\Gamma\lloc\varphi$
says exactly that $\den{\Delta_\Gamma}\subseteq\den{\varphi}$ in every $M$.
Every member of $\Delta_\Gamma$ is derivable from $\Gamma$ by \Cref{lem:GN},
and $\Gamma\subseteq\Delta_\Gamma$ by taking $p=\varepsilon$. By~(S), the two
have the same models, $M\vDash\Gamma$ iff $M\vDash\Delta_\Gamma$.

\begin{lemma}\label{lem:Beasy}
If $\Delta_\Gamma\lloc\varphi$ then $\Gamma\vDash\varphi$.
\end{lemma}

\begin{proof}
Let $M\vDash\Gamma$. Then $M\vDash\Delta_\Gamma$, i.e.\
$\den{\Delta_\Gamma}=M$, by the equivalence just noted. As $\Gamma$ and
$\varphi$ are closed, $\Delta_\Gamma\lloc\varphi$ reads
$\den{\Delta_\Gamma}\subseteq\den{\varphi}$. Hence $\den{\varphi}=M$, i.e.\
$M\vDash\varphi$.
\end{proof}

\begin{lemma}\label{lem:whatDelta}
$\den{\Delta_\Gamma}=\{u\in M\mid\ims{u}\subseteq\den{\Gamma}\}$, the largest
backward-closed subset of $\den{\Gamma}$.
\end{lemma}

\begin{proof}
By \Cref{lem:boxsem}, $\den{\Delta_\Gamma}
=\bigcap_{\gamma,p}\{u\mid\im{p}{u}\subseteq\den{\gamma}\}
=\{u\mid\ims{u}\subseteq\den{\Gamma}\}$. Call that set $T$. Then $T\subseteq\den{\Gamma}$ because $u\in\ims{u}$, and
$\ldn[T]\subseteq T$ because $\ims{v}\subseteq\ims{u}$ whenever
$u\rightsquigarrow v$. Any backward-closed $U\subseteq\den{\Gamma}$
satisfies $\ims{u}\subseteq U\subseteq\den{\Gamma}$ for $u\in U$, so
$U\subseteq T$.
\end{proof}

Thus $\Delta_\Gamma$ is the polyadic form of the $\Box^\omega$-closure of $\Gamma$
\cite[Ex.~1.5.3]{BlackburnDeRijkeVenema2001}, and localizing $\Gamma$ discards
exactly the points from which one can descend out of $\den{\Gamma}$.

%% ================================================================
\section{\texorpdfstring{$\Delta_\Gamma$ localizes $\Gamma$}{Delta-Gamma localizes Gamma}}\label{sec:B}

\subsection*{Idea of the proof}
The classical reduction fails at its final step. In basic modal logic, a
pointed local countermodel is restricted to the submodel generated by its
point. The resulting submodel still satisfies the hypotheses because every
point is reachable from the root, and it still refutes the conclusion because
modal formulas are invariant under generated submodels
\cite[Ex.~1.5.3 and \S2.1]{BlackburnDeRijkeVenema2001}. This invariance argument
is unavailable here. An element variable may be assigned outside the generated
part, and $\exists x$ quantifies over that assignment.

\Cref{lem:locality} characterizes what the language can \emph{see} outside a
backward-closed $C$. An element variable is tested only for equality with the
current point, which lies in $C$, so any value outside $C$ reads as ``not
here''. Backward closure also prevents travel outside $C$. From within $C$,
every point in the complement therefore contributes the same truth value. One
phantom element suffices to represent it.

A matching-logic model, however, cannot contain phantom elements. Every element
must be a legal value of an element variable, a legal witness for $\exists x$,
and, because global consequence demands totality at every point, a point at
which $\Gamma$ holds. \Cref{lem:whatDelta} gives
$C\subseteq\den{\Gamma}$ at the instantiation $C=\ims{w}$ in \Cref{thm:B}.
Thus $\Gamma$ holds at
every point of $C$, and an admissible representative can be formed from a
second copy of $C$. The construction $N=C\times\{0,1\}$ in \Cref{def:N} is a
double cover of $C$. It is not obtained by extracting a submodel from $M$. Its
second sheet is as featureless from the first as
$M\setminus C$ was, while consisting of elements at which $\Gamma$ holds. The
added material must satisfy $\Gamma$, and $C$ is the only available source.
One additional copy suffices because the language can distinguish only one
phantom from within $C$.

The underlying ingredients are classical. There are also classical reasons to
expect that they do not suffice on their own. Generated submodels and their
invariance theorem are standard. Disjoint unions provide the standard witness that the universal
modality is undefinable \cite{GorankoPassy1992}. For neighbouring hybrid
logics, the failure of the modal argument with element variables and binders
explains why strong completeness requires extra apparatus. The available
devices include the satisfaction operator or $A$ \cite{GorankoPassy1992},
infinitary rules, and nominals used as Henkin witnesses. In
\cite{BlackburnTzakova1998}, for example, each existential is witnessed by a
nominal, and the language is expanded with a denumerable reserve of fresh ones.
Here we refine the invariance theorem. \Cref{lem:locality} replaces ``$C$
determines truth on $C$'' with the weaker statement that ``$C$ together with
one bit per variable determines truth on $C$''. We then use doubling to
manufacture globality, a different purpose from the usual definability
counterexample.

This refinement also has a direct precedent. \Cref{lem:locality} is in
substance Proposition~4.4 of \cite{BlackburnSeligman1995}, whose language
shares our variable and positive-arity core, consisting of state variables
bound by $\exists$, modalities of arbitrary arity, and no propositional
variables, ``the purely first-order apparatus''. That proposition proves that
$\exists$-formulas are preserved under proper generated-substructure
isomorphisms. Its accompanying explanation is that $\exists$ can detect
\emph{whether} points lie outside the substructure generated by the point of
evaluation while remaining blind to the information they contain.
Consequently, all non-local generated substructures could be collapsed to a
single point without $\exists$ detecting the difference. Proposition~4.5
performs this collapse by adjoining one isolated dummy point to the generated
part.

This collapse is incompatible with global consequence. A dummy point is
harmless for local satisfaction, which asks only whether it is the current
point. Global consequence $\Gamma\vDash\varphi$ requires $\Gamma$ to be
\emph{total}, that is, satisfied at every point of the model, including the
dummy. An arbitrary isolated point need not satisfy $\Gamma$. The phantom must
therefore be a point at which $\Gamma$ holds, and \Cref{lem:whatDelta} shows
that $C$ is the only available source. The complement is represented by a
second copy of $C$. With no edges between the copies, the second copy is
invisible from the first for the same reason given in
\cite{BlackburnSeligman1995}. It also satisfies $\Gamma$ at every point because
$C$ does. For local satisfaction, the outside can disappear entirely. Under
global consequence, every added point must model $\Gamma$, so the construction
replaces the outside by a copy of the inside.

Removing element variables and $\exists$ identifies the new content of this
construction. \Cref{lem:locality} then reduces to the classical statement that
$C$ determines truth on $C$. The set $C$ is already a countermodel,
\Cref{lem:twocopy} is unnecessary, and one may take $N:=C$. In this case,
\Cref{thm:B} becomes the generated-submodel argument verbatim. The construction
below supplies the additional step required by element variables and $\exists$.

Together with \Cref{sec:boxed}, this argument identifies the difference between
the two localizations. At a point $w$, the condition
$w\in\den{\Delta_\Gamma}$ means that $\Gamma$ holds at every point reachable
from $w$, so $\Delta_\Gamma$ acts as a master box. In
$\Gamma\vDash\varphi$, the premise $M\vDash\Gamma$ requires $\Gamma$ to hold
at every point of $M$, and the conclusion $M\vDash\varphi$ imposes the same
totality requirement on $\varphi$. A universal modality would express this
whole-model condition. The language has no such modality. The difference is
the unreachable part. The two lemmas below show that this part is featureless
and may be replaced by points at which $\Gamma$ holds.

The remaining implication,
$\Gamma\vDash\varphi\Rightarrow\Delta_\Gamma\lloc\varphi$, follows from two
lemmas about a backward-closed set $C\subseteq M$. The locality lemma requires
only backward closure (\Cref{def:bc}). The two-copy lemma also assumes
$\varnothing\ne C\ne M$, as in \Cref{def:N}. Only \Cref{thm:B} instantiates
$C$ as $\ims{w}$ for a point $w$ refuting $\varphi$.
Throughout this section, $C$ is backward closed.

\begin{lemma}[locality]\label{lem:locality}
Let $\rho,\rho':\Var\to M$ satisfy, for every $x$: either
$\rho(x)=\rho'(x)\in C$, or $\rho(x)\notin C$ and $\rho'(x)\notin C$. Then
$\ii{\rho}{\psi}\cap C=\ii{\rho'}{\psi}\cap C$ for every pattern $\psi$.
\end{lemma}

\begin{proof}
Proceed by induction on $\psi$, for all $\rho,\rho'$. If $\psi=x$ is a variable
and $\rho(x)\in C$, then
$\ii{\rho}{x}\cap C=\{\rho(x)\}=\ii{\rho'}{x}\cap C$. If
$\rho(x)\notin C$ then both sides are $\varnothing$. For a constant and for
$\bot$ the denotation does not depend on the valuation. For
$\psi_1\to\psi_2$, intersecting with $C$ turns the complement in $M$ into the
complement in $C$, and the induction hypothesis applies to each side. For
$\sigma(\psi_1,\ldots,\psi_n)$ with $n\ge 1$, backward
closure puts every tuple producing a point of $C$ inside $C$,
where the induction hypothesis applies. Concretely,
$\sigma_M(A_1,\ldots,A_n)\cap C=\sigma_M(A_1\cap C,\ldots,A_n\cap C)\cap C$
for all $A_1,\ldots,A_n$. For
$\exists x.\,\psi_1$ the union over $a\in M$ may be taken termwise, since
$\rho[a/x]$ and $\rho'[a/x]$ again satisfy the hypothesis.
\end{proof}

\begin{definition}\label{def:N}
Assume $\varnothing\ne C\ne M$ from here on. Let $N:=C\times\{0,1\}$. For
$\sigma$ of arity $n\ge 1$,
\[
  \sigma_N\bigl((a_1,i_1),\ldots,(a_n,i_n)\bigr):=
  \begin{cases}
    \bigl(\sigma_M(a_1,\ldots,a_n)\cap C\bigr)\times\{i\},
      & i_1=\cdots=i_n=i,\\
    \varnothing, & \text{otherwise,}
  \end{cases}
\]
where the second case disallows a \emph{mixed} tuple. For a constant $\sigma$,
let $\sigma_N:=(\sigma_M\cap C)\times\{0,1\}$. Fix any
$\out\in M\setminus C$. By
\Cref{lem:locality} nothing below depends on which one. For $i\in\{0,1\}$ let
$\pi_i:N\to M$ keep copy $i$ and send the other copy to $\out$. Thus
$\pi_i(a,i):=a$ and $\pi_i(a,1-i):=\out$ for $a\in C$. For
$\nu:\Var\to N$ write $\pi_i(\nu):=\pi_i\circ\nu:\Var\to M$.
\end{definition}

$N$ is an ordinary model, nonempty because $C$ is, and each $\pi_i(\nu)$ maps
into $M$, so it can be fed directly to the denotation in $M$. A constant has no
arguments, so the mixed case does not arise for it. The same subset
$\sigma_M\cap C$ serves in both copies, vacuously when it is empty. This
symmetry is necessary because assigning the subset to only one copy would break
\Cref{lem:twocopy} at
$\psi=\sigma$ for every constant with $\sigma_M\cap C\ne\varnothing$.

\begin{lemma}[two copies]\label{lem:twocopy}
For every pattern $\psi$ and every $\nu:\Var\to N$,
\[
  \ii{\nu}{\psi}_N
  \;=\;
  \bigl(\ii{\pi_0(\nu)}{\psi}_M\cap C\bigr)\times\{0\}
  \ \cup\
  \bigl(\ii{\pi_1(\nu)}{\psi}_M\cap C\bigr)\times\{1\} .
\]
\end{lemma}

\begin{proof}
Proceed by induction on $\psi$. Fix $i$ and compare the copy-$i$ parts of the
two sides. \emph{Variable.} The copy-$i$ part of $\ii{\nu}{x}_N$ is
$\{(a,i)\}$ if
$\nu(x)=(a,i)$ and $\varnothing$ otherwise, matching
$\ii{\pi_i(\nu)}{x}_M\cap C$, since in the second case
$\pi_i(\nu(x))=\out\notin C$.
\emph{Constant.} The set $\sigma_N$ meets copy $i$ in
$(\sigma_M\cap C)\times\{i\}$.
$\ii{\nu}{\bot}_N=\varnothing$ on both sides.
\emph{Implication.} Complementation in $N$ restricted to copy $i$ is
complementation in $C$, which is the effect of intersecting the $M$-clause with
$C$.
\emph{Symbol, $n\ge 1$.} A value meets copy $i$ only via unmixed tuples from
copy $i$, giving $\bigl(\sigma_M(\bar a)\cap C\bigr)\times\{i\}$ for
$\bar a\in C^n$. Conversely, for $u\in C$, every tuple producing $u$ lies in
$C$ by backward closure. The induction hypothesis converts the two sides.
\emph{Existential.} The union over $(b,j)\in N$ on the left corresponds, by the
induction hypothesis and $\pi_i(\nu[(b,j)/x])=\pi_i(\nu)[\pi_i(b,j)/x]$, to a
union over $e\in C\cup\{\out\}$ on the right.
\Cref{lem:locality} shows that this equals the union over all $e\in M$. For
$e\notin C$ the valuations $\pi_i(\nu)[e/x]$ and $\pi_i(\nu)[\out/x]$ agree
on $C$. A witness realizing $\out$ exists because the other copy is nonempty,
$C$ being so.
\end{proof}

\begin{corollary}\label{cor:NM}
For every closed pattern $\psi$,
$\den{\psi}_N=\bigl(\den{\psi}_M\cap C\bigr)\times\{0,1\}$. Consequently,
$N\vDash\psi$ iff $C\subseteq\den{\psi}_M$.
\end{corollary}

\begin{proof}
For closed $\psi$, $\ii{\pi_i(\nu)}{\psi}_M$ is $\den{\psi}_M$ independently
of $\nu$ and $i$, so \Cref{lem:twocopy} gives the first identity. For the
second, $N=C\times\{0,1\}$, so $\den{\psi}_N$ is all of $N$ exactly
when $\den{\psi}_M\cap C=C$.
\end{proof}

\begin{theorem}[semantic localization]\label{thm:B}
$\Gamma\vDash\varphi \iff \Delta_\Gamma\lloc\varphi$.
\end{theorem}

\begin{proof}
($\Leftarrow$) is \Cref{lem:Beasy}. ($\Rightarrow$) Contrapositively, suppose
$\Delta_\Gamma\not\lloc\varphi$. As $\Delta_\Gamma$ and $\varphi$ are closed,
this gives a model $M$ and a point $w\in M$ with $w\in\den{\Delta_\Gamma}_M$ and
$w\notin\den{\varphi}_M$. Let $C:=\ims{w}$, which is backward closed, contains
$w$, and satisfies $C\subseteq\den{\Gamma}_M$ by \Cref{lem:whatDelta}. If
$C=M$ then $M\vDash\Gamma$ and $M\nvDash\varphi$. Otherwise $\varnothing\ne
C\ne M$, so $N$ is defined, and \Cref{cor:NM} applied to $\Gamma$ and
$\varphi$, both closed, gives $N\vDash\Gamma$ from
$C\subseteq\den{\Gamma}_M$, and $N\nvDash\varphi$ from
$w\in C\setminus\den{\varphi}_M$. Either way $\Gamma\nvDash\varphi$.
\end{proof}

\begin{theorem}[proof-theoretic localization]\label{thm:A}
$\Gamma\vdash\varphi \iff \Delta_\Gamma\lloc\varphi$.
\end{theorem}

\begin{proof}
($\Leftarrow$) By~(L) there is a finite $\Delta_0\subseteq\Delta_\Gamma$ with
$\vdash\bigl(\bigwedge\Delta_0\bigr)\to\varphi$. Each $\delta\in\Delta_0$ is
derivable from $\Gamma$, so finitely many conjunction-introductions and one
modus ponens give $\Gamma\vdash\varphi$.

($\Rightarrow$) By~(S), $\Gamma\vdash\varphi$ gives $\Gamma\vDash\varphi$, and
\Cref{thm:B} turns that into $\Delta_\Gamma\lloc\varphi$.
\end{proof}

%% ================================================================
\section{The composite}\label{sec:composite}

\begin{corollary}[global completeness]\label{cor:main}
$\Gamma\vDash\varphi \iff \Gamma\vdash\varphi$.
\end{corollary}

\begin{proof}
The result follows by composing \Cref{thm:B,thm:A} through
$\Delta_\Gamma\lloc\varphi$.
\end{proof}

\begin{corollary}[conservativity of definedness]\label{cor:cons}
If $\Gamma$ and $\varphi$ are definedness-free, then
$\Gamma\cup\{\lceil x\rceil\}\vdash\varphi$ implies $\Gamma\vdash\varphi$.
\end{corollary}

\begin{proof}
Soundness in the enriched signature gives
$\Gamma\cup\{\lceil x\rceil\}\vDash\varphi$ (\Cref{fig:system} is sound for
every signature). Every definedness-free
$M\vDash\Gamma$ expands to a model of $\Gamma\cup\{\lceil x\rceil\}$, so
$M\vDash\varphi$. Hence $\Gamma\vDash\varphi$, and \Cref{cor:main} applies.
\end{proof}

\begin{remark}[free set variables do not come along]\label{rem:setvars}
One might expect \Cref{cor:main} to extend to patterns with free set variables,
by reading such a variable as a constant. This extension fails under the consequence
relation of \Cref{def:total}. Satisfaction quantifies over all valuations, so a
free set variable is quantified \emph{inside} the premise, whereas a constant
is fixed by the model. Thus the two are not interchangeable. Take $\Gamma=\{X\}$
and $\varphi=\bot$. No model satisfies $X$, since the valuation sending $X$ to
$\varnothing$ makes it not total, so $\Gamma\vDash\bot$ holds vacuously.
Replacing $X$ by a fresh constant $d$ gives $\{d\}\nvDash\bot$, as the model
with $d_M=M$ shows. One direction remains valid under the same reading.
$\Gamma[\bar d/\bar X]\vDash\varphi[\bar d/\bar X]$ implies
$\Gamma\vDash\varphi$, and $\Gamma[\bar d/\bar X]\vdash\varphi[\bar d/\bar X]$
implies $\Gamma\vdash\varphi$ by uniform substitution. The converse direction,
which would transfer \Cref{cor:main}, fails.
\end{remark}

%% ================================================================
\section{Applicative matching logic}\label{sec:aml}

Applicative matching logic \cite{ChenRosu2019AML,ChenRosu2020AML} is the
instance in which $\Sigma$ consists of one binary symbol, \emph{application},
written by juxtaposition, together with any set of constants. It has two
coordinates, which we write $1$ and $2$, and the boxes of \Cref{def:boxes}
become
\[
  \bx{1}\psi = \bigl((\psi\to\bot)\top\bigr)\to\bot,
  \qquad
  \bx{2}\psi = \bigl(\top(\psi\to\bot)\bigr)\to\bot ,
\]
with $a\app b\ldi{1}a$ and $a\app b\ldi{2}b$, and $E^*=\{1,2\}^*$. Everything
above applies unchanged, so \Cref{cor:main} holds for applicative matching
logic. This fragment carries the encodings of sorts, binders and type
systems \cite{ChenRosu2020AML,ChenRosu2020ICFP}.

\begin{remark}[currying is not known to be conservative]\label{rem:curry}
It is tempting to prove \Cref{cor:main} for applicative matching logic alone
and recover the general case by currying,
$\sigma(\varphi_1,\ldots,\varphi_n)^\dagger
 :=(\cdots(c_\sigma\varphi_1^\dagger)\cdots)\varphi_n^\dagger$. One direction
follows directly. Any applicative model of $\Gamma^\dagger$ becomes a model of
$\Gamma$ on the same carrier by reading $\sigma_M$ off the currying chain,
whence $\Gamma\vDash\varphi$ implies
$\Gamma^\dagger\vDash\varphi^\dagger$. The converse transfer of
\emph{derivability} remains unproved. An applicative derivation of
$\varphi^\dagger$ may pass through patterns with no polyadic counterpart,
such as $c_\sigma$ alone or a partial application
$c_\sigma\varphi_1$. Recovering a polyadic derivation therefore requires a
conservativity lemma. Substitution alone does not suffice. The semantic route
is also blocked. Turning a
polyadic countermodel into an applicative one requires manufacturing partial
applications, and the enlarged carrier is visible to the element quantifier,
since $\exists x$ ranges over the added elements and a $\Gamma$ constraining
cardinality distinguishes the two models. Relativizing the quantifier to an
inhabitant predicate repairs this problem, but requires definedness.
Completeness with definedness is already available
\cite[Thm.~15]{ChenRosu2019}. Proving currying conservative in the
definedness-free setting would be a result of independent interest. The present
proof does not rely on such a result.
\end{remark}

%% ================================================================
\section{Adding fixpoints: no effective complete system}\label{sec:mu}

\emph{Matching $\mu$-logic} \cite{ChenRosu2019} extends the patterns of
\Cref{sec:system} with set variables $X\in\SVar$, which a valuation sends to
arbitrary subsets of $M$, and with least fixpoints $\mu X.\,\varphi$ in which
$X$ occurs only positively:
\[
  \ii{\rho}{\mu X.\,\varphi}
  =\bigcap\bigl\{A\subseteq M \;\bigm|\; \ii{\rho[A/X]}{\varphi}\subseteq A\bigr\},
\]
which exists by positivity, with
$\nu X.\,\varphi:=\neg\mu X.\,\neg\varphi[\neg X/X]$. Note that
$\bot=\mu X.X$ and $\top=\nu X.X$, so over a signature the primitives reduce to
$\to$, $\exists$ and $\mu$.

\begin{theorem}\label{thm:mu}
Let $\Sigma=\{\mathsf{s},\oplus,\otimes\}$ with $\mathsf{s}$ unary and
$\oplus,\otimes$ binary, and no constants. The set of valid patterns of
one-sorted definedness-free matching $\mu$-logic over $\Sigma$ is not
recursively enumerable.
\end{theorem}

\begin{corollary}\label{cor:mu}
That fragment has no sound and weakly complete calculus whose proof relation is
recursively enumerable, since enumerating proofs would enumerate the valid
patterns. Therefore no ordinary effectively checkable finite-proof calculus is
even weakly complete for this fragment, and hence none is locally or globally
complete. In particular, $\vDash\varphi\Rightarrow\vdash\varphi$ fails for the
empty theory, as does the analogue of~(L). Thus the completeness premise used
by \Cref{thm:A} is unavailable, and the effective global-completeness question
has a negative answer for this fragment.
\end{corollary}

The construction occupies the rest of this section and gives a many-one
reduction from unsolvability of Diophantine equations over $\mathbb{N}$ to
validity. Write $\oplus,\otimes$ infix. Because the signature has no zero
constant, the free element variable $z$ serves as zero in the encoded
arithmetic. The fixpoint pattern $G$ in \Cref{def:G} marks the elements that
represent numerals and constrains $\mathsf{s},\oplus,\otimes$ to act as
successor, addition, and multiplication on them. In \Cref{def:relativize},
natural-coefficient polynomial terms $p,q$ are used to form a pattern
$\Phi_{p,q}$ that is valid exactly when the equation $p=q$ has no solution in
$\mathbb{N}$.

We leave $z$ free, so validity quantifies universally over its assignments.
This is the one place where we set aside the closure convention of
\Cref{sec:system}. Nothing changes if one prefers the closed
$\forall z.\,\Phi_{p,q}$, since validity already quantifies over every
assignment to $z$.

\subsection*{Comparing arguments at a point}
The absence of $@$ discussed in \Cref{sec:conclusions} prevents shifting the
evaluation point. The \emph{arguments} of a tuple producing the current point
can nevertheless be inspected and compared. In particular, $x\wedge y$ is
nonempty exactly when the two element variables coincide. This comparison
already yields injectivity.

\begin{lemma}\label{lem:muinj}
Let $\mathrm{Inj}:=\forall x\forall y.\,\bigl(\mathsf{s}(x)\wedge\mathsf{s}(y)
\to\mathsf{s}(x\wedge y)\bigr)$. Then $u\in\ii{\rho}{\mathrm{Inj}}$ iff at most
one $a$ satisfies $u\in\mathsf{s}_M(a)$.
\end{lemma}

\begin{proof}
Fix $a,b$. The antecedent holds at $u$ iff $u\in\mathsf{s}_M(a)$ and
$u\in\mathsf{s}_M(b)$. The consequent holds at $u$ iff
$u\in\mathsf{s}_M(v)$ for some $v\in\{a\}\cap\{b\}$, i.e.\ iff $a=b$ and
$u\in\mathsf{s}_M(a)$. So the implication holds for all $a,b$ exactly when $u$
has at most one $\mathsf{s}$-argument.
\end{proof}

\subsection*{Pinning the numerals}
\begin{definition}\label{def:G}
Let $\mathrm{Loc}[X]$ be the conjunction of
\begin{enumerate}[leftmargin=2.9em,itemsep=.15em,label=\textup{(L\arabic*)}]
\item $\mathrm{Inj}$;
\item $\neg\bigl(z\wedge\mathsf{s}(\top)\bigr)$;
\item $\forall x.\,\bigl(x\oplus z\to x\bigr)$;
\item $\forall x\forall y.\,
      \bigl(x\oplus\mathsf{s}(y)\to\mathsf{s}((x\oplus y)\wedge X)\bigr)$;
\item $\forall x.\,\bigl(x\otimes z\to z\bigr)$;
\item $\forall x\forall y.\,
      \bigl(x\otimes\mathsf{s}(y)\to
        (z\wedge x)\vee(((x\otimes y)\wedge X)\oplus x)\bigr)$,
\end{enumerate}
and put $G:=\mu X.\bigl(\mathrm{Loc}[X]\wedge(z\vee\mathsf{s}(X))\bigr)$.
\end{definition}

$X$ occurs only in the consequents of (L4) and (L6) and in $\mathsf{s}(X)$,
always positively, so $G$ is well formed. At a point, (L2) says that if the
point is the zero, then it is not an $\mathsf{s}$-value. Clause (L3) says that
$u\in a\oplus\rho(z)$ implies $u=a$. Clause (L4) says that
$u\in a\oplus\mathsf{s}_M(b)$ implies
$u\in\mathsf{s}_M\bigl((a\oplus b)\cap\ii{\rho}{X}\bigr)$. Clauses (L5) and
(L6) have the analogous readings.

\begin{lemma}[standardness]\label{lem:mustandard}
Fix $M$ and $\rho$, and let $G_M:=\ii{\rho}{G}$. Then:
\begin{enumerate}[leftmargin=2em,itemsep=.15em,label=\textup{(\roman*)}]
\item every $u\in G_M\setminus\{\rho(z)\}$ has exactly one $a$ with
  $u\in\mathsf{s}_M(a)$, and $a\in G_M$; write $a=\mathrm{pred}(u)$;
\item if $G_M\ne\varnothing$ then $\rho(z)\in G_M$, it has no
  $\mathsf{s}$-argument, and iterating $\mathrm{pred}$ from every $u\in G_M$
  reaches $\rho(z)$ in finitely many steps, so $\dep\colon G_M\to\mathbb{N}$
  is well defined, with $\dep(\rho(z))=0$ and
  $\dep(u)=\dep(\mathrm{pred}(u))+1$ for $u\ne\rho(z)$;
\item if $a,b,u\in G_M$ and $u\in a\oplus b$ then $\dep(u)=\dep(a)+\dep(b)$;
\item if $a,b,u\in G_M$ and $u\in a\otimes b$ then
  $\dep(u)=\dep(a)\cdot\dep(b)$.
\end{enumerate}
\end{lemma}

\begin{proof}
The operator
$F(A):=\ii{\rho[A/X]}{\mathrm{Loc}[X]\wedge(z\vee\mathsf{s}(X))}$ is monotone
and $G_M$ is its least fixpoint, so
$G_M=\bigcup_\alpha F^\alpha(\varnothing)$, where
$F^{\alpha+1}(\varnothing):=F(F^\alpha(\varnothing))$ and
$F^\lambda(\varnothing):=\bigcup_{\alpha<\lambda}F^\alpha(\varnothing)$ at
limit stages, and
$G_M=\ii{\rho}{\mathrm{Loc}[G_M]}\cap
 \bigl(\{\rho(z)\}\cup\mathsf{s}_M(G_M)\bigr)$. In particular (L1)--(L6) hold
at every point of $G_M$.

(i) For $u\in G_M$ with $u\ne\rho(z)$, the fixpoint equation gives
$u\in\mathsf{s}_M(a)$ for some $a\in G_M$, and \Cref{lem:muinj} at $u$ gives no
other.

(ii) First, $\rho(z)\notin G_M$ forces $G_M=\varnothing$. The first
approximant satisfies
$F(\varnothing)\subseteq\{\rho(z)\}\cup\mathsf{s}_M(\varnothing)
=\{\rho(z)\}$ and $F(\varnothing)\subseteq G_M$, so $F(\varnothing)$, and with
it every later stage, is empty. So if $G_M\ne\varnothing$ then
$\rho(z)\in G_M$ and (L2) is available at that point. If $\rho(z)$ had an
$\mathsf{s}$-argument, (L2) would fail there. The hypothesis is needed. With
$M=\{0\}$, $\rho(z)=0$, $\mathsf{s}_M(0)=\{0\}$ and $\oplus,\otimes$ empty,
(L2) fails at $0$, so $G_M=\varnothing$ while $\rho(z)$ does have an
$\mathsf{s}$-argument.
Each $u$ enters the approximation at a least stage, and by the fixpoint
equation $\mathrm{pred}(u)$ enters strictly earlier, so the iteration descends
and must halt at the only element without a predecessor.

(iii) Induction on $\dep(b)$. If $\dep(b)=0$ then $b=\rho(z)$ by (ii), and
(L3) at $u$ gives $u=a$. If $\dep(b)=n+1$, put
$b':=\mathrm{pred}(b)\in G_M$, so $b\in\mathsf{s}_M(b')$ and $\dep(b')=n$.
Then $u\in a\oplus\mathsf{s}_M(b')$, so (L4) at $u$ gives some
$v\in(a\oplus b')\cap G_M$ with $u\in\mathsf{s}_M(v)$. By (L2) at $u$ we have
$u\ne\rho(z)$, so $\mathrm{pred}(u)=v$ by (i) and $\dep(u)=\dep(v)+1$. The
induction hypothesis gives $\dep(v)=\dep(a)+n$.

(iv) Induction on $\dep(b)$. If $\dep(b)=0$ then $b=\rho(z)$ and (L5) at $u$
gives $u=\rho(z)$, so $\dep(u)=0$. If $\dep(b)=n+1$ with
$b':=\mathrm{pred}(b)$, then (L6) at $u$, taken with $x\mapsto a$ and
$y\mapsto b'$, leaves two cases. If $u$ lies in the first disjunct then
$u\in\{\rho(z)\}\cap\{a\}$, so $u=a=\rho(z)$ and
$\dep(u)=0=\dep(a)\cdot\dep(b)$. Otherwise there is
$v\in(a\otimes b')\cap G_M$ with $u\in v\oplus a$. The induction hypothesis
gives $\dep(v)=\dep(a)\cdot n$, and (iii), applicable because
$v,a,u\in G_M$, gives $\dep(u)=\dep(v)+\dep(a)=\dep(a)\cdot(n+1)$.
\end{proof}

\subsection*{The reduction}
Let $p,q$ be terms over $z,\mathsf{s},\oplus,\otimes$ in variables
$\bar y=y_1,\ldots,y_k$, taken to be exactly the variables occurring in $p$ or
$q$. These are polynomials with natural coefficients, with numerals written
$\mathsf{s}^n(z)$. We use two terms because subtraction is unavailable, and
any Diophantine equation can be rearranged into this form.

\begin{definition}\label{def:relativize}
Conjoin $G$ at every node: $z^G:=z\wedge G$, $y_i^G:=y_i\wedge G$,
$(\mathsf{s}\,t)^G:=\mathsf{s}(t^G)\wedge G$,
$(t_1\oplus t_2)^G:=(t_1^G\oplus t_2^G)\wedge G$, and likewise for $\otimes$.
Put $\Phi_{p,q}:=\neg\,\exists y_1\cdots\exists y_k.\,(p^G\wedge q^G)$.
\end{definition}

Since symbols propagate $\varnothing$, one node escaping $G$ empties the whole
term.

\begin{lemma}\label{lem:mureduction}
$\vDash\Phi_{p,q}$ if and only if $p(\bar n)=q(\bar n)$ has no solution
$\bar n\in\mathbb{N}^k$.
\end{lemma}

\begin{proof}
($\Leftarrow$) Suppose there is no solution and, towards a contradiction, that
$u\in\ii{\rho}{\exists\bar y.\,(p^G\wedge q^G)}$ for some $M,\rho,u$. Then
there are $\bar b$ with $u$ in both denotations under $\rho[\bar b/\bar y]$.
Every node of $p^G$ and $q^G$ is conjoined with $G$, so every value computed at
every node, including each $b_i$ and $u$, lies in $G_M$. Thus $G_M$ is
nonempty, and the clause of \Cref{lem:mustandard}(ii) applies. Induction on
term structure, using \Cref{lem:mustandard}(iii),(iv) at the $\oplus$ and
$\otimes$ nodes and (i),(ii) at the $\mathsf{s}$ nodes and at $z$, gives
$\dep(u)=p(\dep(\bar b))$ and $\dep(u)=q(\dep(\bar b))$, so $\dep(\bar b)$
solves the equation.

($\Rightarrow$) Suppose $p(\bar n)=q(\bar n)$. Take $M:=\mathbb{N}$ with
$\mathsf{s}_M(k):=\{k+1\}$, $k\oplus_M l:=\{k+l\}$,
$k\otimes_M l:=\{k\cdot l\}$, and $\rho(z):=0$, $\rho(y_i):=n_i$. We claim $G_M=\mathbb{N}$.
Each of (L1)--(L6) holds at every point when $X$ is interpreted as
$\mathbb{N}$, and $\{0\}\cup\mathsf{s}_M(\mathbb{N})=\mathbb{N}$, so
$\mathbb{N}$ is a fixpoint of $F$. For leastness we check
$F^{m+1}(\varnothing)=\{0,\ldots,m\}$ by induction on $m$. At $m=0$ the
candidates are $\{0\}\cup\mathsf{s}_M(\varnothing)=\{0\}$. Clause (L4) is vacuous
there, since $0\notin a\oplus_M\mathsf{s}_M(b)$ for any $a,b$, and (L6) holds
because $0\in a\otimes_M\mathsf{s}_M(b)$ forces $a=0$, when the disjunct
$z\wedge x$ contains $0$. For $m\ge 1$ the candidates are
$\{0\}\cup\mathsf{s}_M(\{0,\ldots,m-1\})=\{0,\ldots,m\}$, so only the clauses
at the new point $m$ need attention, the others holding by the induction
hypothesis. (L1), (L2), (L3) and (L5) are immediate. For (L4),
$m\in a\oplus_M\mathsf{s}_M(b)$ forces $a+b=m-1$, which lies in the previous
stage, and the consequent then contains $\mathsf{s}_M(m-1)=\{m\}$. For (L6),
$m\in a\otimes_M\mathsf{s}_M(b)$ forces $a(b+1)=m$ with $a\ge 1$, so
$ab=m-a<m$ lies in the previous stage and the second disjunct contains
$ab+a=m$. Hence each $t^G$ denotes $\{t(\bar n)\}$ and
$p^G\wedge q^G$ is nonempty at $p(\bar n)$, so $\Phi_{p,q}$ is not total there.
\end{proof}

\begin{ednote}
Removing the occurrences of $X$ from (L4) and (L6) permits the recursion
equations to produce intermediate values outside $G$, where nothing constrains
them. For addition, injectivity identifies the intermediate as the predecessor,
so the occurrence of $X$ in (L4) is dispensable and
\Cref{lem:mustandard}(iii) still holds. Multiplication offers no analogous
recovery. Removing $X$ from (L6) therefore breaks
\Cref{lem:mustandard}(iv). There the fixpoint variable supplies what definedness
would otherwise supply.

The disjunct $z\wedge x$ in (L6) serves a separate purpose. Omitting it causes
the construction to fail outright. Multiplication by zero returns zero, so at
the point $\rho(z)$ the clause taken with $x=z$ would demand that
$(z\otimes y)\wedge X$ be inhabited, that is, that $\rho(z)$ already lie in
$X$. The clause thus makes $\rho(z)$ depend on itself. In the intended model on
$\mathbb{N}$, this is fatal. At zero, the antecedent holds with $x=z$, while the
recursive consequent is empty at the first stage, so $F(\varnothing)=\varnothing$
and $G_M=\varnothing$, and the witness required by \Cref{lem:mureduction} is
lost.
The extra disjunct breaks the circularity by treating the zero product as a
base case. Addition needs no counterpart, because a successor is never zero.
\end{ednote}

\begin{proof}[Proof of \Cref{thm:mu}]
The map $(p,q)\mapsto\Phi_{p,q}$ is computable, so \Cref{lem:mureduction} is a
many-one reduction of $U:=\{(p,q)\mid p=q \text{ has no solution in }
\mathbb{N}\}$ to the set of valid patterns. We check that $U$ is
$\Pi^0_1$-complete. Its complement $S$, the set of pairs
$(p,q)$ for which $p=q$ has a solution, is recursively enumerable by dovetailing
over tuples. For hardness, fix a $\Sigma^0_1$-complete $K$. By the
Matiyasevich--Robinson--Davis--Putnam theorem \cite{Matiyasevich1970} there is
an integer polynomial $P(e,\bar x)$ with $e\in K$ iff
$\exists\bar x\in\mathbb{N}^m.\,P(e,\bar x)=0$. Writing $P=P^{+}-P^{-}$ with
$P^{+},P^{-}$ of natural coefficients and substituting the numeral for $e$ gives
computable natural-coefficient terms $p_e,q_e$ with $e\in K$ iff
$\exists\bar x.\,p_e(\bar x)=q_e(\bar x)$. So $S$ is $\Sigma^0_1$-complete and
$U$ is $\Pi^0_1$-complete. Were the valid patterns recursively enumerable, the
reduction would make $U$ recursively enumerable as well, a contradiction.
\end{proof}

\begin{remark}[the method cannot be run in the modal $\mu$-calculus either]
\label{rem:muloc}
\Cref{thm:mu} closes the middle and right of the bottom row of
\Cref{fig:landscape}. Our method does not apply to the left entry, for a
different and milder reason. The modal $\mu$-calculus is not compact. For
example, the set
\[
  \Delta=\{\mu X.(p\vee\Diamond X)\}\cup\{\neg p,\ \neg\Diamond p,\
  \neg\Diamond\Diamond p,\ \ldots\}
\]
is unsatisfiable, since the fixpoint asserts that $p$ is reachable in finitely
many steps while the remaining formulas deny it at every finite distance, yet
every finite subset is satisfiable. So $\Delta\lloc\bot$ while no finite
$\Delta_0\subseteq\Delta$ has $\vdash\bigwedge\Delta_0\to\bot$, and the
analogue of~(L) fails. \Cref{thm:A}($\Leftarrow$) is therefore unavailable
there too. This causes no obstacle because global consequence is expressible by
$\nu X.(\bigwedge\Gamma\wedge\Box X)$ for finite $\Gamma$ and needs no
localization.

That the $\mu$-calculus is non-compact appears to be folklore. An infinitary
axiomatization was already given alongside the finitary one in
\cite{Kozen1983}, as a non-compact logic requires for strong
completeness, and \cite{AmblerEtAl1995} proves completeness for both. The
corresponding statement for propositional dynamic logic is explicit in
\cite{RenardelKooiVerbrugge2008}. The example records the status of the three
notions in \Cref{sec:system}. Weak completeness holds, local completeness
fails, and global completeness holds only for finite $\Gamma$.
\end{remark}

\subsection*{How small the fragment is}
\Cref{thm:mu} uses one sort and three symbols of arity at most two. Definedness,
equality, membership, and constants are absent. The fragment has element
variables with $\forall$ and $\exists$, together with a single fixpoint template
$G$ and no $\nu$. Since $G$ is conjoined at every node of $p^G$ and $q^G$, the
expanded pattern carries one $\mu$-binder occurrence per node. These occurrences
are unnested copies of one template, so the alternation depth is one. Each copy
binds a single set variable with three positive occurrences in its body, and
none are free in $\Phi_{p,q}$. Finally, the theory is empty. This fact is
essential to \Cref{cor:mu}. Because the statement concerns validity, it refutes
weak completeness directly without relying on a specially chosen $\Gamma$.

Two restrictions recover this paper's neighbouring positive results. Removing
$\mu$ yields \Cref{cor:main}. Removing element variables and $\exists$ yields
the modal $\mu$-calculus, whose axiomatization is complete
\cite{Kozen1983,Walukiewicz2000}. The combination of $\mu$, element variables,
and $\exists$ yields arithmetic.

\begin{remark}[what \Cref{thm:mu} adds to what was known]\label{rem:mudef}
The result that matching $\mu$-logic has no effective sound and complete proof
system is due to \cite[Prop.~22--23]{ChenRosu2019}. Using definedness, that
argument gives a \emph{finite} theory $\Gamma^{\mathbb{N}}$ (functionality, no
confusion, the recursion equations for $+$ and $\times$, and the inductive
domain axiom $\mu D.\,0\vee\mathsf{succ}(D)$), all of whose models are
isomorphic to $(\mathbb{N},+,\times)$, so that $\Gamma^{\mathbb{N}}\vDash\varphi$
is not recursively enumerable. Enumerate this finite theory as
$\Gamma^{\mathbb{N}}=\{\gamma_1,\ldots,\gamma_n\}$, and for each arithmetic
sentence $\sigma$ let $\sigma^*$ denote its translation into matching logic.
Since $\Gamma^{\mathbb{N}}$ is finite and definedness internalizes totality,
$\lfloor\gamma\rfloor$ being total exactly when $\gamma$ is, that argument also
yields the weak form
$\vDash\bigl(\bigwedge_i\lfloor\gamma_i\rfloor\bigr)\to\sigma^*$ iff
$\mathbb{N}\vDash\sigma$, and $\mathrm{Th}(\mathbb{N})$ is not recursively
enumerable.

The same pattern recurs elsewhere. As shown in
\cite[Rem.~37]{ChenLucanuRosu2026}, where initial algebras are axiomatized in
matching logic (no-junk by a least fixpoint, no-confusion, and equality modulo
$E$), non-ground equational validity in initial algebras is
$\Pi^0_2$-complete, so no effective sound and complete system exists for that
class either. That argument, too, runs through definedness, which their
theories import for equality, and through a nonempty theory.

Both arguments rely on the fact that \emph{definedness internalizes totality and
$\mu$ turns internalized totality into induction}. \Cref{thm:mu} operates
without definedness and with no hypotheses. The occurrences of $X$ in (L4) and
(L6) allow the fixpoint variable itself to perform the internalization. The
negative result therefore extends to the definedness-free column of
\Cref{fig:landscape}, which \cite{ChenRosu2019} leaves untouched.
\end{remark}

\begin{remark}[what is not minimized]\label{rem:munotmin}
Three aspects remain unminimized. \emph{Arity.} The symbols $\oplus$ and
$\otimes$ are binary because addition
and multiplication are ternary relations, which a binary symbol expresses at a
point. Whether unary symbols suffice remains open, although $\mathrm{Loc}[X]$
is in effect a relativized universal quantifier and could plausibly constrain
an encoded triple structure. \emph{Applicative matching logic.} Currying
$\mathsf{s},\oplus,\otimes$ into constants under a single application is
expected to go through, since \Cref{lem:muinj} applies to each argument
position of application, but the constraints then quantify over partial
applications and we have not written them. \emph{Conjunctivity.} See below.
\end{remark}

\begin{remark}[the aconjunctive fragment is untouched]\label{rem:acon}
A fixpoint $\mu X.\,\varphi$ is \emph{aconjunctive} when no conjunction inside
$\varphi$ has $X$ free in two or more conjuncts. The pattern $G$ of
\Cref{def:G} fails this condition. Its body is
$\mathrm{Loc}[X]\wedge(z\vee\mathsf{s}(X))$, with $X$ free in both conjuncts.
This feature is essential. Removing the occurrences of $X$ from (L4) and (L6)
makes the body aconjunctive and is exactly what breaks
\Cref{lem:mustandard}(iv). Thus the occurrences of $X$ that enable the encoding
also make the body conjunctive.

This matters because the completeness proof of \cite{Kozen1983} is for the
aconjunctive fragment and is elementary. Completeness of the full calculus was
established in \cite{Walukiewicz2000}. The elementary argument is therefore available in the
fragment beyond the reach of our obstruction. This makes completeness for the
aconjunctive fragment of matching $\mu$-logic a well-scoped target. Two caveats
remain. First, failure of \emph{this} reduction does not establish completeness,
since arithmetic might be encodable aconjunctively by some other device.
Second, aconjunctivity is defined for the modal $\mu$-calculus, so transposing it
here requires care because $\forall x.\,\varphi$ is a negated existential that
may introduce conjunctions. In its favour, the fixpoints arising in language
semantics are mostly reachability formulas $\mu X.(\varphi\vee\sigma(X))$,
which are directly aconjunctive.
\end{remark}

%% ================================================================
\section{The four regimes}\label{sec:regimes}

The two localization identities are semantic localization,
$\Gamma\vDash\varphi\iff\Delta_\Gamma\lloc\varphi$, and proof-theoretic
localization, $\Gamma\vdash\varphi\iff\Delta_\Gamma\lloc\varphi$. For semantic
localization, $\Delta_\Gamma\lloc\varphi\Rightarrow\Gamma\vDash\varphi$ is
direct. The reverse implication uses the double-cover construction of
\Cref{sec:B}. For proof-theoretic localization,
$\Gamma\vdash\varphi\Rightarrow\Delta_\Gamma\lloc\varphi$ follows from
soundness and semantic localization. The reverse implication uses strong local
completeness~(L). The outcomes below depend on which of these identities or
directions remains available.

\begin{center}
\small
\begin{tabular}{@{}lll@{}}
\toprule
\textbf{Setting} & \textbf{Status} & \textbf{Relevant step or obstruction} \\
\midrule
one sort, fixpoint-free
  & complete
  & both identities hold; \Cref{cor:main} \\[.3em]
many sorts, fixpoint-free
  & \textbf{fails}
  & the sorts of $\Gamma$ cannot feed the sort of \\
  & & the conclusion (\Cref{prop:manysorted}) \\[.3em]
with definedness
  & complete
  & \Cref{thm:B} degenerates: $\lfloor\cdot\rfloor$ is $A$, and \\
  & & $\Delta_\Gamma$ may be taken to be
      $\{A\gamma\mid\gamma\in\Gamma\}$ \\[.3em]
with $\mu$
  & \textbf{fails}
  & \Cref{thm:A}($\Leftarrow$) fails: it consumes~(L), \\
  & & and no effective calculus has~(L) there (\Cref{thm:mu}) \\
\bottomrule
\end{tabular}
\end{center}

One qualification applies to the whole table. When we say that completeness
\emph{fails}, we mean that it fails for the system of
\Cref{fig:system}, since a counterexample exhibits $\Gamma\vDash\varphi$ with
$\Gamma\nvdash\varphi$. That is weaker than non-axiomatizability. Indeed,
global consequence for the definedness-free system, with any number of sorts,
translates into ordinary first-order consequence in the usual way, so for an
effectively presented countable signature and a recursively enumerable theory it
is recursively enumerable and \emph{some} sound and complete recursive system
for it exists. The counterexample shows that the system considered here is not
such a system. The $\mu$ row is different. By \Cref{thm:mu}, validity itself is
not recursively enumerable, so no sound calculus with a recursively enumerable
proof relation is weakly complete there, whatever axioms one adopts. This result
does not address the semantic relation or calculi with undecidable proof
relations.

\subsection*{The many-sorted system}
\Cref{sec:system} fixes one sort, so we define here the many-sorted system used
in the second row. This is the original system of \cite{ChenRosu2019}, whose
one-sorted restriction appears in \Cref{fig:system}. We include the definition
to make this section self-contained.

A many-sorted signature is a set $S$ of sorts together with a set $\Sigma$ of
symbols, each carrying an arity $s_1\cdots s_n\to t$ with $s_i,t\in S$. Element
variables are sorted, $\Var_s$ for each $s$, and patterns carry a sort:
\[
  \varphi ::= x \mid \sigma(\varphi_1,\ldots,\varphi_n)
    \mid \varphi_1\to\varphi_2 \mid \bot_s \mid \exists x{:}s.\,\varphi,
\]
Here $x\in\Var_s$ has sort $s$. The pattern
$\sigma(\varphi_1,\ldots,\varphi_n)$ has sort $t$ when
$\sigma:s_1\cdots s_n\to t$ and each $\varphi_i$ has sort $s_i$. The implication
$\varphi_1\to\varphi_2$ is formed only from patterns of a common sort and
inherits that sort. The pattern $\bot_s$ has sort $s$. Finally,
$\exists x{:}s.\,\varphi$ has the sort of $\varphi$, whatever $s$ is, so
quantification does not change the sort of a pattern. Write
$\mathrm{sort}(\varphi)$ for that sort. The derived connectives
of \Cref{sec:system} are taken at each sort, and $\forall x{:}s$ abbreviates
$\neg\exists x{:}s.\,\neg$ as before.

A model provides a nonempty carrier $M_s$ for every $s\in S$ and, for every
$\sigma:s_1\cdots s_n\to t$, a map
$\sigma_M:M_{s_1}\times\cdots\times M_{s_n}\to\Pw(M_t)$, extended pointwise to
sets as in \Cref{sec:system} and again empty as soon as one argument is. A
valuation sends each $x\in\Var_s$ into $M_s$, and
$\ii{\rho}{\varphi}\subseteq M_{\mathrm{sort}(\varphi)}$ is given by the clauses
of \Cref{sec:system}, with $\ii{\rho}{\bot_s}=\varnothing$, with complementation
taken in $M_{\mathrm{sort}(\varphi)}$, and with the union in the existential
clause taken over $a\in M_s$. Totality is at the pattern's own sort, so
\[
  M\vDash\varphi
  \iff \ii{\rho}{\varphi}=M_{\mathrm{sort}(\varphi)}\ \text{for every }\rho,
\]
and $M\vDash\Gamma$, $\vDash\varphi$ and $\Gamma\vDash\varphi$ are defined as
in \Cref{def:total}.

An application context carries two sorts, a hole sort and a result sort:
\[
  C ::= \Box_s \mid \sigma(\varphi_1,\ldots,C,\ldots,\varphi_n),
\]
where the second form requires $C$ to have result sort $s_i$ when it occupies
argument position $i$ of $\sigma:s_1\cdots s_n\to t$. The resulting context has
sort $t$. A context with hole sort $s$ and result sort $t$ therefore exists only
when some chain of argument positions leads from $s$ to $t$.

The proof system is \Cref{fig:system} read with sorts. The propositional
schemes, modus ponens and the two $\exists$ rules are taken at each sort, with
$x\in\Var_s$ in the quantifier rules. Propagation and framing are taken at each
symbol and each argument position. For each such rule, the premise has the sort
of that position and the conclusion has the result sort of the symbol. Existence is
$\exists x{:}s.\,x$ at each sort. The singleton axiom is taken for contexts
$C_1,C_2$ with a common result sort and holes of the sort of $x$. Derivations
are finite, and their only $\Gamma$-dependent lines are members of $\Gamma$.
Soundness is the implication from $\Gamma\vdash\varphi$ to
$\Gamma\vDash\varphi$ and is proved in \cite[Thm.~13]{ChenRosu2019}. Input~(S)
of \Cref{sec:system} is its one-sorted case.

\subsection*{The many-sorted failure}
The second row requires a separate argument because \Cref{def:boxed} is stated
for one sort. The following form does not require a many-sorted localization.

\begin{proposition}\label{prop:manysorted}
Let the sorts be $b,a,c$ with symbols $f:b\to a$ and $g:b\to c$, and put
\[
  \Gamma=\{\forall x{:}b\,\forall y{:}b.\,f(x\wedge y)\},
  \qquad
  \varphi=\forall x{:}b\,\forall y{:}b.\,(g(x)\leftrightarrow g(y)).
\]
Then $\Gamma$ is satisfiable, $\Gamma\vDash\varphi$, and $\Gamma\nvdash\varphi$.
\end{proposition}

\begin{proof}
\emph{Satisfiability.} Take $M_b$ a singleton, $M_a$ and $M_c$ singletons, and
$f$ and $g$ total.

\emph{Entailment.} Let $M\vDash\Gamma$. If $M_b$ had distinct $r,s$, the
valuation $\rho$ with $\rho(x)=r$ and $\rho(y)=s$ would give
$\ii{\rho}{x\wedge y}=\varnothing$, and symbols propagate $\varnothing$, so
$\ii{\rho}{f(x\wedge y)}=\varnothing$, which is not $M_a$ because carriers are
nonempty. Hence $M_b$ is a singleton, every valuation sends $x$ and $y$ to the
same element, and $\ii{\rho}{g(x)\leftrightarrow g(y)}=M_c$ for every $\rho$.
So $M\vDash\varphi$.

\emph{Non-derivability.} Say that a sort $s$ \emph{feeds} a sort $t$ if $s=t$
or some symbol has an argument of sort $s$ and result sort $t$. Write
$s\Rightarrow t$ for the reflexive transitive closure. In this signature no
symbol has an argument of sort $a$, so $a\Rightarrow t$ holds only for $t=a$.

Every rule of \Cref{fig:system} has premises whose sorts feed the sort of its
conclusion. Modus ponens and the quantifier rules are sort-preserving, and
framing puts a pattern of sort $s$ into one argument position of a symbol, so
its conclusion has that symbol's result sort $t$, and such a symbol has an
argument of sort $s$. So if the conclusion of a rule has sort $t\ne a$, then no
premise of it has sort $a$.

By induction on derivations, every line whose sort is not $a$ is derivable from
the axioms alone. A hypothesis line has sort $a$ and is excluded, an axiom line
needs nothing, and a rule with conclusion of sort $t\ne a$ has, by the previous
paragraph, no premise of sort $a$. The induction hypothesis therefore applies
to all of its premises. Since $\varphi$ has sort $c$, from
$\Gamma\vdash\varphi$ we would get $\vdash\varphi$ and hence
$\vDash\varphi$ by~(S). But $\varphi$ is not valid. Take
$M_b=\{r,s\}$ with $g_M(r)\ne g_M(s)$. Hence $\Gamma\nvdash\varphi$.
\end{proof}

\Cref{prop:manysorted} establishes that the one-sort hypothesis is essential.
The sentence in \Cref{sec:boxed} about words in $E^*$ identifies where the
hypothesis enters the proof. With one sort every word composes, whereas here no
word leads from the sort of $\Gamma$ to the sort of the conclusion. The
definedness case is complete for an independent reason. A universal modality
makes localization directly available \cite{GorankoPassy1992}, which is why
that case was proved first \cite[Thm.~15]{ChenRosu2019} and why it provides no
evidence for the definedness-free setting. In the $\mu$ case, the identified
failure lies in oracle~(L). This argument does not locate a failure in the
construction of \Cref{sec:B}.
\Cref{thm:A}($\Leftarrow$) uses~(L), and by \Cref{thm:mu} no
effective calculus supplies it. We do not analyse whether
\Cref{lem:locality,lem:twocopy} survive $\mu$ and $\nu$, and nothing here depends
on that question. See \Cref{sec:mu}.

\begin{corollary}\label{cor:noencoding}
There is no translation from many-sorted definedness-free matching logic into
\emph{any} one-sorted definedness-free matching logic that preserves global
consequence and reflects derivability.
\end{corollary}

\begin{proof}
Preservation applied to the counterexample, then \Cref{cor:main} in the target,
then reflection, would give $\Gamma\vdash\varphi$.
\end{proof}

The corollary leaves intact forward preservation of derivability, which is the
direction needed for proof export. It rules out transporting completeness back
along an encoding. In particular the sorted and parametric encodings of
\cite{ChenRosu2020AML,ChenRosu2020ICFP}, which route sorted quantification
through definedness, cannot be replaced by a definedness-free translation that
simultaneously preserves global consequence and reflects derivability.

%% ================================================================
\section{What localization can reach}\label{sec:limits}

In \Cref{sec:B} we proved \Cref{thm:A}($\Rightarrow$) by composing soundness
with \Cref{thm:B}. This concise proof uses only one property of the proof
system, which we now state explicitly. The resulting formulation applies to
other calculi and extends the nominal obstruction from \Cref{fig:system} to
every system satisfying that property.

Throughout this section $\Gamma$ is a set of closed patterns, and we abbreviate
the property $\Delta_\Gamma\lloc\varphi$ by $\mathrm{I}(\varphi)$. Unfolding
\Cref{def:boxed}, $\mathrm{I}(\varphi)$ says that
$\den{\Delta_\Gamma}_M\subseteq\ii{\rho}{\varphi}_M$ for every model $M$ and
every valuation $\rho$.

\begin{definition}[respecting localization]\label{def:respects}
A rule \emph{respects localization} if, for every set $\Gamma$ of closed
patterns, whenever its side conditions are met and $\mathrm{I}$ holds of each
premise, $\mathrm{I}$ holds of the conclusion.
\end{definition}

\begin{theorem}\label{thm:obstruction}
Let $\vdash$ be a derivability relation over the syntax of \Cref{sec:system}
generated by well-founded derivations whose leaves are members of $\Gamma$ or
valid patterns, and each of whose rule applications respects localization in the
sense of \Cref{def:respects}.
Then $\Gamma\vdash\varphi$ implies $\Delta_\Gamma\lloc\varphi$. Consequently,
if $\Gamma\vDash\varphi$ while $\Delta_\Gamma\not\lloc\varphi$, then
$\Gamma\nvdash\varphi$, so $\vdash$ is not globally complete.
\end{theorem}

\begin{proof}
Induction on the derivation. If $\varphi\in\Gamma$ then
$\varphi\in\Delta_\Gamma$, taking $p=\varepsilon$, so
$\den{\Delta_\Gamma}\subseteq\den{\varphi}$. If $\varphi$ is a valid leaf then
$\ii{\rho}{\varphi}=M$ by validity, and $\den{\Delta_\Gamma}\subseteq M$.
The remaining case is an application of a rule, and it is exactly
\Cref{def:respects}. The consequence follows by contraposition.
\end{proof}

The axioms are unconstrained beyond validity. Adding any valid schemes leaves
the conclusion unchanged because a valid pattern denotes the whole carrier,
which contains $\den{\Delta_\Gamma}$. The obstruction therefore lies entirely
in the rules, and a given system can be checked rule by rule.

\begin{proposition}\label{prop:rulesrespect}
The rules of \Cref{fig:system} respect localization, and so do the rules of the
standard hybrid Hilbert systems: modus ponens, necessitation, generalization,
and the rule \textup{(Name)}, from $i\to\varphi$ with $i$ a nominal occurring
neither in $\varphi$ nor in $\Gamma$, infer $\varphi$.
\end{proposition}

\begin{proof}
The proof uses two properties of $\den{\Delta_\Gamma}$ established in
\Cref{sec:boxed}. It is independent of the valuation because every member of
$\Delta_\Gamma$ is closed, and it is backward closed by
\Cref{lem:whatDelta}.

\emph{Modus ponens.} For each $\rho$, the result follows from
$\den{\Delta_\Gamma}\subseteq\ii{\rho}{\psi}\cap\ii{\rho}{\psi\to\varphi}
\subseteq\ii{\rho}{\varphi}$.

\emph{Framing, and with it necessitation.} This is the case treated in
\Cref{sec:B}: if $u\in\den{\Delta_\Gamma}$ lies in the denotation of the framed
antecedent, witnessed by a tuple whose $i$-th component is $a$, then
$u\ldi{e}a$, so $a\in\den{\Delta_\Gamma}$ by backward closure and the
hypothesis applies at $a$. For $\Box$, if $u\in\den{\Delta_\Gamma}$ and
$u\ldn v$, then $v\in\den{\Delta_\Gamma}\subseteq\ii{\rho}{\varphi}$, so
$u\in\ii{\rho}{\Box\varphi}$.

\emph{Generalization, and with it $\exists$-generalization.} Suppose
$\den{\Delta_\Gamma}\subseteq\ii{\rho}{\varphi}$ for every $\rho$. Since
$\den{\Delta_\Gamma}$ does not depend on the valuation, the same inclusion
holds at every $\rho[a/x]$, so
$\den{\Delta_\Gamma}\subseteq\bigcap_{a\in M}\ii{\rho[a/x]}{\varphi}
=\ii{\rho}{\forall x.\,\varphi}$. For the $\exists$ form, suppose
$\den{\Delta_\Gamma}\subseteq\ii{\rho}{\varphi_1\to\varphi_2}$ for every
$\rho$, with $x\notin\mathrm{FV}(\varphi_2)$, and let
$u\in\den{\Delta_\Gamma}\cap\ii{\rho}{\exists x.\,\varphi_1}$. Then
$u\in\ii{\rho[a/x]}{\varphi_1}$ for some $a$, and the hypothesis at
$\rho[a/x]$, together with $u\in\den{\Delta_\Gamma}$, gives
$u\in\ii{\rho[a/x]}{\varphi_2}=\ii{\rho}{\varphi_2}$.

\emph{\textup{(Name)}.} Let $u\in\den{\Delta_\Gamma}_M$ and let $M'$ agree with
$M$ except that $i$ is interpreted as $\{u\}$. The nominal $i$ does not occur in
$\Gamma$, hence not in $\Delta_\Gamma$, so
$\den{\Delta_\Gamma}_{M'}=\den{\Delta_\Gamma}_M\ni u$. The hypothesis at $M'$
gives $u\in\ii{\rho}{i\to\varphi}_{M'}$, and $u\in\den{i}_{M'}$, so
$u\in\ii{\rho}{\varphi}_{M'}$. Since $i$ does not occur in $\varphi$ either,
$\ii{\rho}{\varphi}_{M'}=\ii{\rho}{\varphi}_M$, and $u$ lies in it.
\end{proof}

\subsection*{Two applications}
The first is a restatement of what \Cref{sec:regimes} records, now free of any
particular calculus.

\begin{corollary}\label{cor:msobstruction}
Let $\vdash$ be a sound derivability relation over the signature of
\Cref{prop:manysorted}, generated by well-founded derivations in which the only
$\Gamma$-dependent initial lines are members of $\Gamma$, all other axioms are
valid and independent of $\Gamma$, the rule schemes and their side conditions
are independent of $\Gamma$, and every premise has a sort that feeds the sort of
the conclusion. Then $\vdash$ is not globally complete.
\end{corollary}

\begin{proof}
Those are exactly the properties of \Cref{fig:system} that the
non-derivability half of \Cref{prop:manysorted} uses.
\end{proof}

The last hypothesis alone is far from enough. Taking $\vdash$ to be $\vDash$
itself, presented by the single $\Gamma$-dependent axiom scheme ``$\psi$,
whenever $\Gamma\vDash\psi$'', gives a sound and globally complete relation
whose rules have no premises at all, so the sort condition holds vacuously.

This criterion differs from \Cref{def:respects}. We do not know whether the
many-sorted failure follows from that definition. The two negative results use
different methods.

The second concerns $\mathrm{H}(\forall)$, which is basic matching logic with
\emph{nominals} added. These constants $i$ have singleton interpretations
$i_M$ and therefore name points of the model instead of arbitrary subsets.
Everything in \Cref{sec:boxed} goes through verbatim for that language, since
nominals are closed patterns.

\begin{corollary}\label{cor:hforall}
Let $\vdash$ satisfy the hypotheses of \Cref{thm:obstruction}, read over
$\mathrm{H}(\forall)$. Then $\vdash$ is not globally complete. Explicitly, with
one nominal $i$,
\[
  \{i\}\vDash\forall x.\,x
  \qquad\text{while}\qquad
  \{i\}\nvdash\forall x.\,x .
\]
\end{corollary}

\begin{proof}
Both patterns are closed. If $M\vDash i$ then $\den{i}=M$, and $\den{i}$ is a
singleton, so $M$ has one point $w$. In that model,
$\ii{\rho}{\forall x.\,x}=\bigcap_{a\in M}\{a\}=\{w\}=M$, giving
$M\vDash\forall x.\,x$, which is the left-hand relation.

For the right-hand one, take $M=\{a,b\}$ with $i_M=\{a\}$ and every symbol
interpreted as $\varnothing$, so that $\ldn=\varnothing$. Every box is then
vacuously total, so $\den{\bx{p}i}=M$ for $p\ne\varepsilon$ and
$\den{\Delta_{\{i\}}}=\den{i}=\{a\}$, while
$\ii{\rho}{\forall x.\,x}=\{a\}\cap\{b\}=\varnothing$. Hence
$a\in\den{\Delta_{\{i\}}}\setminus\ii{\rho}{\forall x.\,x}$ and
$\Delta_{\{i\}}\not\lloc\forall x.\,x$. Now apply \Cref{thm:obstruction}.
\Cref{prop:rulesrespect} shows, in particular, that the standard hybrid rules
satisfy its rule hypothesis.
\end{proof}

The counterexample identifies the additional constraint that nominals allow
$\Gamma$ to impose. A nominal is total exactly when the
carrier has one point, so $\{i\}$ constrains the \emph{cardinality} of the
model. Recall that the double-cover construction starts with a model $M$ and a
backward-closed core $C\subseteq M$ satisfying $\varnothing\ne C\ne M$, then
forms $N=C\times\{0,1\}$. This use of $C$ is separate from the
application-context metavariable in the many-sorted system. Because $C$ is
nonempty, $N$ has at least two points and cannot be a one-point model.
Definedness-free theories can still express properties of the carrier as a
whole. For example,
$\forall x.\,x$ is total exactly on one-point models, and the entailment above
uses that fact. Such a theory, however, cannot survive \Cref{cor:NM}, whose
construction doubles the carrier. That additional room supports the doubling
argument of \Cref{sec:B}.

\begin{remark}[what a complete system would have to look like]\label{rem:whatcomplete}
Global consequence for $\mathrm{H}(\forall)$ remains axiomatizable. Under the
standard translation, nominals become constants and global consequence becomes
ordinary first-order consequence. For an effectively presented countable
signature and a recursively enumerable theory, it is therefore recursively
enumerable, and some sound and complete recursive system exists.
\Cref{cor:hforall} requires any such well-founded system whose leaves are
hypotheses or valid patterns to contain a rule that does \emph{not} respect
localization. Such a rule must see beyond the component generated by the point
of evaluation. An infinitary rule could qualify if its premise family enforces
a genuinely global condition, as could a rule with an explicitly global side
condition. Infinitely many premises alone are insufficient. This requirement
is more informative than a bare non-existence claim would have been, although
we do not know an explicit example of such a system.
\end{remark}

\begin{remark}[maximality]\label{rem:maximality}
Together, \Cref{cor:main,cor:msobstruction,cor:hforall} place basic unsorted
matching logic at a boundary. The basic logic is globally complete, while two
extensions within the definedness-free setting defeat a broad class of
calculi. With sorts, the sorts of $\Gamma$ need not feed the sort of the
conclusion. With nominals, $\Gamma$ can fix the cardinality of the carrier. The
first result uses sort flow, and the second uses localization.
\Cref{sec:mu} gives a third failure with a different and stronger conclusion.
\end{remark}

%% ================================================================
\section{Conclusions, future work, and open problems}\label{sec:conclusions}

\subsection*{Relation to modal and to hybrid logic}
The \emph{expressivity} boundary underlying this note already appears in
\cite{BlackburnSeligman1995}. Their Theorem~4.1 and Proposition~4.11 place
$\exists$ strictly below the target of the standard translation, namely
first-order logic with equality over one relation symbol per modality. Adding
the ``somewhere'' modality, which is definedness, reaches that target.
\Cref{cor:main} and \Cref{sec:regimes} show that the same boundary also
separates global completeness from its failure.

The classical reduction and its failure here were sketched in
\Cref{sec:intro}. Two further points clarify the boundary. When element
variables and $\exists$ are removed, our argument reduces exactly to the
classical one, as noted in \Cref{sec:B}. The construction addresses precisely
those two additional features. On the other side,
\cite{ArecesTenCate2007} observes that
$\mathrm{H}(@,\downarrow)$ sentences \emph{are} invariant under generated
submodels, provided the generated part is taken to include the points named by
nominals, so for that language the classical route does go through. It is the
quantifier over state variables that destroys the invariance. The same feature
makes $\exists x$ range over the discarded part in our setting.

Under the equivalence of \cite{LeusteanMoangaSerbanuta2019}, one-sorted basic
matching logic is the nominal-free fragment of $\mathrm{H}(\forall)$ with a
polyadic family of modalities, without the satisfaction operator or the
universal modality. Thus \Cref{cor:main} is a global strong completeness theorem
for that logic. Its distinction from the known results is the \emph{global}
consequence relation.
Hybrid languages with binders are
axiomatized in \cite{BlackburnTzakova1998} with respect to \emph{local}
consequence, blending canonical models with witnessed maximal consistent sets,
and \cite[Thm.~16]{ChenRosu2019} is described there as a generalization of that
proof for $\mathrm{H}(\forall)$ to many-sorted polyadic signatures. That theorem is our
oracle~(L) and serves as an input here. We have not found the corresponding
global statement. The $\Box^\omega$ reduction requires the invariance just
discussed. The transfer of \cite{GorankoPassy1992} requires $A$. The strong
completeness results for pure extensions of $\mathrm{H}(@)$ and
$\mathrm{H}(@,\downarrow)$ require $@$ together with the non-orthodox (Name)
and (Paste) rules \cite{BlackburnTenCate2006}. The remaining results use
infinitary rules. We would welcome references to a global result.

Two differences from the setting of \cite{BlackburnTzakova1998} point in
opposite directions. That language has neither $@$ nor $A$, so it lies on the
definedness-free side with ours, and the question below is genuinely the global
one for it. We first fix terminology to avoid a recurring source of confusion.
A hybrid signature has three kinds of atom:
propositional variables, \emph{nominals}, which the model interprets as
singletons and which no binder can bind, and \emph{state variables}, which the
assignment interprets and which $\forall$ binds. Matching logic's element
variables are the state variables. Its nullary symbols are the propositional
variables. Below, \emph{nominal} always means the model-fixed constant.

Nominals have a matching-logic counterpart on the definedness side of the
dichotomy. A nullary symbol denotes an arbitrary subset, but definedness allows
that subset to be constrained to a singleton. Adding the axiom
$\exists x.\,\lfloor c\leftrightarrow x\rfloor$, which makes $c$ a
\emph{functional} pattern in the sense of
\cite{Rosu2017}, forces $c_M$ to be a singleton in every model of the
theory, exactly as required for a nominal. Definedness is needed for a
\emph{uniform} definition because no pattern of the definedness-free language
denotes $\lceil\varphi\rceil$ in every model. A theory can nevertheless achieve
the same effect with an auxiliary symbol. Take an ordinary unary $r$ and
write $A_r(\psi):=\neg r(\neg\psi)$. The axiom $\forall x.\,r(x)$ forces
$r_M(a)=M$ for every $a$, hence $r_M(B)=M$ for $B\ne\varnothing$ and
$r_M(\varnothing)=\varnothing$, so $r$ behaves as definedness and $A_r$ as
totality. Adding $\exists y.\,A_r(c\leftrightarrow y)$ then forces $c_M$ to be a
singleton, and the resulting theory has models on every nonempty carrier.

This construction does not conflict with \Cref{cor:NM}, whose hypothesis it
defeats. Under
$\forall x.\,r(x)$ every point is $r$-related to every point, so the only
nonempty backward-closed set is $M$ and no proper core exists. \Cref{cor:NM}
says that whenever a model of $\Gamma$ \emph{does} have a backward-closed $C$
with $\varnothing\ne C\ne M$, the model $N$ also satisfies $\Gamma$ and
interprets $c$ as $(c_M\cap C)\times\{0,1\}$, which is never a singleton. So a
theory can define nominal behaviour in the definedness-free language only by
axioms that remove the proper cores on which our construction runs. This is the
theory-relative form of the language-level dichotomy. Polyadic modalities
introduce only a routine extension. That language, however, also
carries \emph{model-fixed nominals} alongside bindable state variables, whereas
ours does not. The obstruction above shows that the construction of
\Cref{sec:B} fails after nominals are added. Thus \Cref{cor:main} is global and
nominal-free, while their result is local and includes nominals. Open Problem~1
asks whether global completeness and nominals can be combined, so it concerns
their language.

\subsection*{Models, frames, and the known incompleteness theorems}
We have worked throughout at the level of models, with \Cref{def:total}
defining model-level consequence. A frame level could also be defined because a
symbol $\sigma_M:M^n\to\Pw(M)$ is an $(n+1)$-ary relation and a nullary symbol is
a propositional valuation. Our semantics does not hold the positive-arity
relational structure fixed while quantifying separately over all valuations of
the nullary symbols. The programming-language interpretation in
\Cref{sec:intro} explains this choice. The symbols of $\Sigma$ belong to the
syntax of the language, and models, understood as implementations, assign their
meaning. At the frame level, the positive-arity symbols would remain fixed as
relations while the nullary symbols varied as valuations. Such quantification
would vary the meanings of the language's constants. A program property should
not be required to hold under every reinterpretation of those constants. The
standard translation of \cite[\S10]{Rosu2017} therefore places global
consequence inside first-order logic. As noted in \Cref{sec:regimes}, model
consequence is recursively enumerable for an effectively presented countable
signature and a recursively enumerable theory. The negative entries in that
section concern the system of \Cref{fig:system}. They do not establish
non-axiomatizability.

This result is compatible with the known incompleteness theorems for normal
modal logics \cite{Thomason1974,vanBenthem1978}, surveyed in
\cite[\S4.4]{BlackburnDeRijkeVenema2001}, because those theorems concern frames. A
frame validates $\varphi$ when \emph{every} valuation makes $\varphi$ true
everywhere, so frame validity translates to
$\forall P_1\cdots\forall P_k\,\forall y.\,ST_y(\varphi)$, which is $\Pi^1_1$
rather than first-order. Validity over the class of \emph{all} frames is still
decidable for the basic modal language, but consequence over a constrained frame
class is a second-order notion. \cite{Thomason1975} reduces second-order
consequence to global frame consequence, showing how far the latter can be from
recursively enumerable. Kripke incompleteness is the failure of a
normal modal logic to coincide with the modal theory of any class of frames,
which concerns the relation between a substitution-closed set of formulas and a
frame class. We instead study whether a calculus captures a model-level
consequence relation.

The two levels differ with respect to uniform substitution. Formulas valid on a
frame are closed under substitution for propositional letters. Formulas valid
on a model need not be closed because the model fixes their interpretations.
Matching logic's $\vdash$ from $\Gamma$ is deliberately not closed under such
substitutions. A language definition fixes the meanings of its symbols, so
substitution for those symbols is generally unsound. The hybrid languages of
\Cref{fig:landscape} behave similarly. Their nominals are fixed by the model,
not by the assignment, and are therefore part of the model data.

\subsection*{Why the construction is available here}
Element variables denote singletons selected by the \emph{valuation}, whereas
the model does not fix them. This choice reflects the needs of language
semantics for program variables, bound variables, and symbolic values.
Proof-theoretic considerations did not motivate it. This choice is exactly what
makes \Cref{sec:B} possible. In the notation of that construction, $M$ is the
original model, $C\subseteq M$ is a nonempty proper backward-closed core, and
$N=C\times\{0,1\}$ is the double-cover model. A nullary symbol denotes an
arbitrary subset, so $\sigma_N:=(\sigma_M\cap C)\times\{0,1\}$ is legal because
the bare semantics does not require a nullary symbol to denote a singleton. An
appropriate theory may impose singletonhood, as noted in
\Cref{sec:conclusions}. Such axioms leave
no proper backward-closed core on which the construction can operate. A hybrid
nominal $i$ is a model-fixed constant with $i_M=\{p\}$, and every legal
interpretation $i_N$ must also be a singleton. If $p\in C$, the right-hand side
of the two-copy identity in \Cref{lem:twocopy} for $\psi=i$ contains both
$(p,0)$ and $(p,1)$, whereas $i_N$ is a singleton. If $p\notin C$, that
right-hand side is empty, whereas $i_N$ must still be a singleton. Thus the
identity fails in either case. The double-cover construction is therefore
available in the nominal-free state-variable fragment of
$\mathrm{H}(\forall)$. It fails in $\mathrm{H}$, $\mathrm{H}(@)$, and
$\mathrm{H}(\forall)$ as presented in the hybrid-logic literature because all
three languages contain model-fixed nominals. In $\mathrm{H}(@)$, the
satisfaction operator supplies point evaluation directly, so its completeness
argument does not need the double-cover reduction.

\subsection*{The satisfaction operator, and two regimes rather than three}
The satisfaction operator of hybrid logic evaluates a formula at the point
named by a variable, from wherever one is standing:
$M,g,w\vDash @_x\varphi$ iff $M,g,g(x)\vDash\varphi$. The evaluation point does
not occur on the right, so $@_x\varphi$ is point-independent. The operator
internalizes the satisfaction relation in hybrid logic. Labelled deduction of
the form ``at world $i$, $\varphi$'' becomes object-level, Henkin constructions
go through, and every pure axiomatic extension of $\mathrm{H}(@)$ is complete
with the non-orthodox (Name) and (Paste) rules
\cite{BlackburnTenCate2006,ArecesTenCate2007}.

In matching logic, definedness provides the same operator:
\[
  @_x\varphi \;:=\; \lceil x\wedge\varphi\rceil \;=\; \lfloor x\to\varphi\rfloor ,
\]
Both patterns are total exactly when $\rho(x)\in\ii{\rho}{\varphi}$. This is
the membership construct $x\in\varphi$, which is standard matching logic
\cite{Rosu2017,ChenLucanuRosu2021}. Its identification with the hybrid
satisfaction operator, together with the equivalence between matching logic
with definedness and $\mathrm{H}_\Sigma(@,\forall)$, is due to
\cite{LeusteanMoangaSerbanuta2019}. That work also records that $A$ becomes
definable as $\forall x.\,@_x\varphi$ once $@$ is present.

This definability result explains a structural difference. Hybrid logic has
three levels: $\mathrm{H}$, $\mathrm{H}(@)$, and the extension with $A$. The
middle level remains distinct because nominals are constants. With finitely
many nominals, $@$ reaches only named points and cannot express $A$. Matching
logic has no corresponding middle level because
\[
  \lceil\varphi\rceil \;=\; \exists x.\,@_x\varphi
\]
and element variables are quantifiable, so $@$ and definedness are
interdefinable. \emph{Matching logic has two regimes where hybrid logic has
three}. The same use of variables in place of constants makes the double cover
work. It also makes the definedness-free/definedness dichotomy intrinsic to the
present setting.

The middle hybrid level retains decidability. $\mathrm{H}(@)$ over arbitrary
frames is decidable,
whereas the binder already makes $\mathrm{H}(\forall)$ undecidable, and $@$ on
top of it reaches full first-order expressive power
\cite{BlackburnSeligman1995,ArecesTenCate2007}. The middle level therefore
supports algorithmic reasoning, including satisfiability checking and model
checking, which is unavailable with quantifiable variables. For language
semantics, this loss imposes no additional cost because program properties are
already undecidable. The relevant requirement is deduction from a theory, which
\Cref{cor:main} and
\cite[Thm.~15]{ChenRosu2019} supply.

At this level, expressiveness is described by an equivalence. Through
\cite{LeusteanMoangaSerbanuta2019}, unsorted definedness-free matching logic
\emph{is} the nominal-free state-variable fragment of polyadic
$\mathrm{H}(\forall)$, without $@$ or $A$. Symbols correspond to polyadic
modalities, element variables to state variables, and nullary symbols to
propositional atoms. Nominals correspond to nullary symbols axiomatized to be
functional, which requires definedness or a theory that supplies its effect.
No additional expressiveness is claimed or needed at this level. The difference concerns
the questions being asked, with matching logic focusing on global consequence
from a theory. Additional strength appears only with fixpoints, where matching
$\mu$-logic reaches first-order logic with least fixpoints and hence
non-axiomatizability.

Finally, $@$ is \emph{not} definable in the definedness-free fragment. This
follows from the locality argument used above for definedness. With $\rho(x)$
outside $C$, the pattern
$@_x\varphi$ reports on a point that \Cref{lem:locality} says no
definedness-free pattern can see. No internal definition is therefore
available, so \Cref{sec:B} supplies a model-theoretic substitute.

\subsection*{Mechanization, and an open challenge}
Mechanized developments exist for matching logic and for hybrid logic. The
proof of \Cref{cor:main} is therefore a natural next target. On the
matching-logic side, \cite{BereczkyEtAl2022}
formalizes matching logic in Coq, and \cite{ChevalMacovei} formalizes
applicative matching logic in Lean at the Institute for Logic and Data Science,
with export of Metamath proof objects. On the hybrid side, \cite{Oltean2023}
formalizes the syntax, semantics, Hilbert system and soundness of
$\mathrm{H}(\forall)$ in Lean~4, following \cite{BlackburnTzakova1998} and
leaving completeness open. \cite{Ericson2026} has since proved completeness.
The two obstacles were the construction of fresh names for the root witnessed
maximal consistent set and, separately, for the witnessed
$\Diamond$-successor. Most directly related,
\cite{OlteanMacoveiLeustean2026} formalizes a many-sorted hybrid polyadic modal
logic in Lean, following \cite{LeusteanMoangaSerbanuta2019}, with a
machine-checked soundness theorem and a domain-specific language for many-sorted
signatures.

Related work in Isabelle/HOL includes \cite{FromBlackburnVilladsen2020}, which
formalizes a Seligman-style tableau system for basic hybrid logic.

\emph{We propose the mechanization of \Cref{cor:main} as an open challenge},
whether in Coq, Lean, Isabelle/HOL, or another system. The argument is
independent of any particular assistant or development above, and we welcome
approaches from directions not represented by those developments. Three entry
points of increasing ambition are available:
\begin{enumerate}[leftmargin=1.9em,itemsep=.2em,label=(\roman*)]
\item mechanize \Cref{lem:locality,lem:twocopy} alone. This consists of two
  five-case structural inductions, uses no fixpoints or canonical model, and
  captures the entire model-theoretic content of the paper.
\item mechanize \Cref{cor:main} with (L) and (S) assumed. This adds only the
  composition of \Cref{thm:A,thm:B}.
\item discharge (L) as well, yielding a self-contained machine-checked global
  completeness theorem for matching logic.
\end{enumerate}
The absence of a canonical-model construction makes this task tractable. The
model refuting $\varphi$ is defined explicitly in \Cref{def:N} from data
provided by the oracle, while a canonical model is ordinarily the most
expensive part of mechanizing a completeness proof.

The challenge also differs from the existing formalizations. The system of
\cite{OlteanMacoveiLeustean2026} carries both $@$ and $\forall$, hence $A$.
It therefore lies on the definedness side of the dichotomy above, where
completeness was already available \cite[Thm.~15]{ChenRosu2019}. To our
knowledge, the definedness-free system studied here has not been mechanized.

\subsection*{Open problems}
\begin{enumerate}[leftmargin=1.9em,itemsep=.35em]
\item \emph{A native globally complete system for $\mathrm{H}(\forall)$.}
  Pulling a complete first-order calculus back along the standard translation
  already gives a complete system. The open problem is to find a \emph{native},
  syntax-directed calculus (Hilbert, sequent, tableau, or labelled) for global
  consequence in $\mathrm{H}(\forall)$ and to identify a minimal rule or
  structural principle that necessarily violates localization
  (\Cref{cor:hforall,rem:whatcomplete}). This is a hybrid-logic question.
\item \emph{The many-sorted boundary.} \Cref{prop:manysorted} shows that the
  one-sort hypothesis cannot simply be dropped, even for satisfiable $\Gamma$.
  A condition on the pair $(\Gamma,\varphi)$ that is necessary and sufficient
  for the localization argument to survive across sorts remains unknown. A
  natural candidate requires every sort occurring in $\varphi$ to be reachable
  along chains of argument positions from a sort occurring in $\Gamma$.
\item \emph{Conservativity of currying.} Is the translation of
  \Cref{rem:curry} derivability-reflecting in the definedness-free setting? A
  positive answer would make the general case a corollary of applicative
  matching logic and yield the preferred presentation.
\item \emph{An intermediate regime.} Is there a useful matching-logic analogue
  of $\mathrm{H}(@)$, with definedness present but out of reach of the
  quantifier, so that $\lceil\varphi\rceil=\exists x.\,@_x\varphi$ is
  blocked? Does it recover any of the decidability available at the middle
  hybrid level? The question is whether the collapse to two regimes is forced
  by the semantics or only by the syntax.
\item \emph{Fixpoints.} By \Cref{thm:mu} there is no effective complete
  system, so \Cref{thm:A}($\Leftarrow$) has no oracle to consume. The
  obstruction uses
  a conjunctive fixpoint, leaving the aconjunctive fragment as the natural
  target (\Cref{rem:acon}), and the minimal arity and the applicative case
  open (\Cref{rem:munotmin}).
\end{enumerate}

\end{document}